\documentclass[11pt]{article}
\def\withcolors{1}
\def\withnotes{0}

\usepackage{ccanonne}
\usepackage{booktabs}
\usepackage{palatino}
\title{Private Distribution Testing with Heterogeneous Constraints:\\
Your Epsilon Might Not Be Mine}
\author{Cl\'ement L. Canonne\thanks{University of Sydney. Email: \email{clement.canonne@sydney.edu.au}.} \and Yucheng Sun\thanks{ETH Z\"urich. Email: \email{yucsun@student.ethz.ch}.}}

\begin{document}
\maketitle
\begin{abstract}
    Private closeness testing asks to decide whether the underlying probability distributions of two sensitive datasets are identical or differ significantly in statistical distance, while guaranteeing (differential) privacy of the data. As in most (if not all) distribution testing questions studied under privacy constraints, however, previous work assumes that the two datasets are \emph{equally} sensitive, i.e., must be provided the same privacy guarantees. This is often an unrealistic assumption, as different sources of data come with different privacy requirements; as a result, known closeness testing algorithms might be unnecessarily conservative, ``paying'' too high a privacy budget for half of the data.
    In this work, we initiate the study of the closeness testing problem under \emph{heterogeneous} privacy constraints, where the two datasets come with distinct privacy requirements.
    
    We formalize the question and provide algorithms under the three most widely used differential privacy settings, with a particular focus on the \emph{local} and \emph{shuffle} models of privacy; and show that one can indeed achieve better sample efficiency when taking into account the two different ``epsilon'' requirements.
\end{abstract}
\section{Introduction}
    \label{sec:introduction}
Hypothesis testing allows a statistician, practitioner, or scientist to validate their model or to detect whether one of their assumptions is statistically improbable. One of the prototypical hypothesis testing tasks is two-sample goodness-of-fit, which asks to determine whether two unknown probability distributions are equal, based on samples from both. This task has received a lot of attention from the computer science community over the past decades as part of the broader area of \emph{distribution testing}, where testing questions are phrased as promise problems with a particular emphasis on finite-sample guarantees and data-efficient algorithms (see, \eg~\cite{Rubinfeld:12:Survey,Canonne:15:Survey,CanonneTopicsDT22}, \cite[Chapter~11]{Goldreich:17}, and references within). In distribution testing, two-sample goodness-of-fit corresponds to \emph{closeness testing}, where the two unknown distributions $\p,\q$ are over a discrete domain of size $\ab$. Given a distance parameter $\dst\in(0,1]$, one then seeks to distinguish (with high probability) between the cases $\p=\q$ and $\totalvardist{\p}{\q}>\dst$ from as few samples from $\p,\q$ as possible, where $\dtv$ denotes the total variation (statistical) distance. A long line of work in the distribution testing literature culminated in a full understanding of the sample complexity (\ie the number of observations required) of this question, with respect to all parameters~\cite{BFRSW:00,Valiant:11,ChanDVV14,DiakonikolasGKPP21,CanonneS22}.

The recent rise of privacy concerns, and with it a focus on privacy-preserving data analysis, led researchers to consider the natural question of \emph{private} hypothesis testing, where the samples are seen as sensitive data, and the testing algorithms must guarantee differential privacy (DP)~\cite{DMNS:06}, a rigorous definition of privacy which has become the \emph{de facto} standard. Various distribution testing questions have been studied under this lens, including that of closeness testing, for which tight bounds on the sample complexity have been established in the so-called ``central'' model of differential privacy~\cite{ADR:17,ASZ:18:DP,Zhang21}. Yet, these results suffer from one major shortcoming: they assume both sets of samples (the one coming from $\p$, and the other from $\q$) to be \emph{equally sensitive}, that is, to require the same privacy guarantees (the parameter $\priv>0$ of differential privacy, where smaller values of $\priv$ correspond to better privacy guarantees).

This assumption, while justifiable in some settings, is misguided in many others: for instance, when the two datasets come from different companies, demographic groups, or even countries subject to different legal requirements -- a typical use case for closeness testing, where one seeks to check if the data from two distinct populations have similar statistical properties. Another extreme use case would be when only one of the datasets has privacy constraints, and the other is from a simulated process (\eg when checking if the distribution of the output from a digital twin, or of synthetic data, matches that of the real world). In these cases, using the ``same $\priv$'' for data from both distributions would be unnecessarily conservative, and could lead to requiring much more (costly, or hard to either gather or generate) data from one distribution than required.

To address this shortcoming, we initiate the study of closeness testing under \emph{heterogeneous} privacy constraints, where the two datasets come with distinct privacy parameters $\priv_1,\priv_2>0$. We further formulate the question not only in the (central) DP model, but also in two others of the most widely used distributed models of privacy, the more stringent \emph{local} DP model~\cite{KLNRS:11,DuchiW:13} and the \emph{shuffle} model of privacy~\cite{CheuSUZZ19,ErlingssonFMRTT19}. To the best of our knowledge, our work is also the first to formulate closeness testing in the two latter models, even for homogeneous privacy constraints ($\priv_1=\priv_2$). 
We next elaborate on our results and detail our contributions.

\subsection{Our Results and Contributions}

Our first contribution is to formalize the question of closeness testing under heterogeneous privacy constraints in three models of differential privacy: the local model, the shuffle model, and the central model (\cref{ssec-models}). While this formalization is somewhat straightforward in the latter case, it is less so in the first two, especially in the shuffle model. Indeed, the shuffle model of privacy relies at its core on a trusted channel (the ``shuffler'') which randomly permutes the messages from all users, effectively anonymizing them. 
However, the very question of closeness testing requires the ability to distinguish between the messages from two groups of users: the ones holding samples from $\p$, and the ones holding samples from $\q$. Thus, some care has to be given in how to define the problem in the first place, putting an emphasis on what choice captures the best the possible use cases discussed earlier. Our definition for the shuffle privacy setting involves two ``shufflers'' (one per group of users), and the privacy guarantee applies to the shuffled output of each of them separately. This aligns with the practical setting where the two datasets come from users from two different company or entities, in which case the shuffling is performed in an early stage (``between'' the users and the corresponding company), \emph{before} being sent out in the world.

Other definitional choices could have been (1)~to add to each message a (non-private) label, to specify from which population they came, and use a single random permutation for all messages; or (2)~to restrict the type of permutations (no longer uniformly random) to only shuffle the messages ``within each population.'' While these two options are equivalent to the one we chose, they are less intuitive, possibly more cumbersome to analyze, and obscure the original motivation for the problem.\medskip

Our second contribution is to provide algorithms achieving non-trivial trade-offs between the two privacy parameters $\priv_1,\priv_2$ in all three models of privacy considered, showing that it \emph{is} possible balance the different privacy requirements of the two populations to do significantly better than defaulting to $\priv=\min(\priv_1,\priv_2)$. Moreover, our results are the first (even for homogeneous privacy constraints) for closeness testing in the local and shuffle models of privacy, and yield sample-optimal (up to constant factors) bounds in the former.

In order to state our results, we first recall the distinction between \emph{private-coin} and \emph{public-coin} distributed protocols:\footnote{Confusingly, the ``private'' in ``private-coin'' does not refer to differential privacy, but to the fact that the randomness (``coin'') of a user is hidden from all others.} in the former, each user has only access to their own randomness, independent of every other user's. In the latter, however, there exists a common random seed (\emph{in addition} to each user's personal randomness), publicly available to all parties (users, analyzer, and world) but still independent of the users' data. Thus, while in both cases the protocols are non-interactive, in public-coin protocols this common random seed can be used to achieve better accuracy, by letting the users somehow coordinate.\footnote{We note that while much of the work in the shuffle model (and, slightly less so, in the local model) does not focus on this distinction, we do so here as the availability of a common random seed is known to make a difference in related testing problems, such as uniformity and identity testing, in both the local and shuffle models of privacy~\cite{AcharyaCFST21,AJM:20,BCJM:20,CanonneL22} as well as in other (non-private) distributed settings~\cite{AcharyaCT:IT1,ACT:19}.} For detailed definitions of the privacy settings and types of distributed protocols, we refer the reader to~\cref{sec:preliminaries}.

Our first results address closeness testing under heterogeneous \emph{local} privacy constraints, establishing a tight trade-off in all parameters.
\begin{theorem}[Local Privacy, Private-Coin]
    \label{theo:ldp:private:coin:intro}
There exists a \emph{private-coin} protocol for closeness testing which guarantees $\priv_1$-local privacy to the $\ns_1$ users of the first group, and $\priv_2$-local privacy to the $\ns_2$ users of the second, with $\ns_1 = \bigO{\frac{\ab^{3/2}}{\priv_1^2\dst^2}}$ and $\ns_2 = \bigO{\frac{\ab^{3/2}}{\priv_2^2\dst^2}}$. Moreover, this is optimal.
\end{theorem}

\begin{corollary}[Local Privacy, Public-Coin]
    \label{theo:ldp:public:coin:intro}
There exists a \emph{public-coin} protocol for closeness testing which guarantees $\priv_1$-local privacy to the $\ns_1$ users of the first group, and $\priv_2$-local privacy to the $\ns_2$ users of the second, with $\ns_1 = \bigO{\frac{\ab}{\priv_1^2\dst^2}}$ and $\ns_2 = \bigO{\frac{\ab}{\priv_2^2\dst^2}}$. Moreover, this is optimal.
\end{corollary}

Our next two results concern the shuffle model of privacy, with algorithms guaranteeing (approximate) differential privacy. For simplicity, we only provide here an informal statement, omitting the at most logarithmic dependence on the parameter $\privdelta$ and focusing on $\priv_1,\priv_2$. We refer the reader to~\cref{theo:sdp_ns,theo:sdp_public_ns} for the full statements.
\begin{theorem}[Shuffle Privacy, Private-Coin (Informal)]
    \label{theo:shuffle:private:coin:intro}
There exists a \emph{private-coin} protocol for closeness testing which guarantees $\priv_1$-shuffle (approximate) privacy to the $\ns_1$ users of the first group, and $\priv_2$-shuffle (approximate) privacy to the $\ns_2$ users of the second, with 
\[
    \ns_1 = \bigO{\max\Paren{\frac{\ab^{1/2}}{\dst^2}, \frac{\ab^{3/4}}{\priv_1\dst}}}
\]
and $\ns_2 = \bigO{\frac{\priv_1^2}{\priv_2^2}\ns_1}$ (assuming without loss of generality that $\priv_2 \leq \priv_1)$. 
\end{theorem}

\begin{corollary}[Shuffle Privacy, Public-Coin (Informal)]
    \label{theo:shuffle:public:coin:intro}
There exists a \emph{public-coin} protocol for closeness testing which guarantees $\priv_1$-shuffle (approximate) privacy to the $\ns_1$ users of the first group, and $\priv_2$-shuffle (approximate) privacy to the $\ns_2$ users of the second, with 
\[
    \ns_1 = \bigO{\max\Paren{\frac{\ab^{1/2}}{\dst^2}, \frac{\ab^{1/2}}{\priv_1\dst}, \frac{\ab^{2/3}}{\priv_1^{2/3}\dst^{4/3}}}}
\]
and $\ns_2 = \bigO{\frac{\priv_1^2}{\priv_2^2}\ns_1}$ (assuming without loss of generality that $\priv_2 \leq \priv_1)$. 
\end{corollary}

Interestingly, specializing our results to the homogeneous privacy constraints case ($\priv_1=\priv_2$), our upper bounds for local and shuffle privacy match\footnote{That is, exactly match the optimal sample complexity in the locally private case; and match the best known upper bounds in the shuffle private case for approximate privacy (or for pure privacy if one treats $\privdelta$ as a constant), and nearly match the corresponding lower bounds for shuffle DP.} the best known algorithms for the simpler problem of \emph{identity} testing (where one of the two distributions is fully known in advance). This shows, perhaps surprisingly, that unlike in the non-private and central DP cases, there is no sample complexity gap between closeness and identity testing.\medskip 

Finally, we provide a simple closeness testing algorithm for the \emph{central} model of differential privacy:
\begin{theorem}[Central Privacy]
    \label{theo:central:intro}
There exists an algorithm for closeness testing which guarantees $\priv_1$-differential privacy to the $\ns_1$ users of the first group, and $\priv_2$-differential privacy to the $\ns_2$ users of the second, with 
\[
    \ns_1 = \bigO{\max\Paren{\frac{\ab^{1/2}}{\dst^2}, \frac{\ab^{1/2}}{\priv_1^{1/2}\dst}, \frac{\ab^{2/3}}{\dst^{4/3}}, \frac{\ab^{1/3}}{\priv_1^{2/3}\dst^{4/3}}, \frac{1}{\priv_1\dst}}}
\]
and $\ns_2 = \bigO{\frac{\priv_1}{\priv_2}\ns_1}$ (assuming without loss of generality that $\priv_2 \leq \priv_1)$. 
\end{theorem}

Our results can be interpreted in two ways. The first focuses on privacy as a \emph{requirement} from the two groups of users, and looks at how the costs varies among the two populations. That is, our results quantifies how much more data one needs to collect from the group with more stringent privacy requirements, compared to the group of less ``privacy-demanding'' users. Our upper bounds show that the overhead scales at most quadratically with the ratio of privacy parameters, \ie as $(\priv_1/\priv_2)^2$ under local and shuffle privacy, and as $\priv_1/\priv_2$ under (central) differential privacy.

The second point of view focuses on privacy as a \emph{promise} (or incentive) instead of a requirement: under this lens, our results show that if more users from one group are willing to participate in the analysis, or if one of the two populations is larger, then it can automatically be guaranteed better data privacy (and our algorithms provide a bound on ``how much more privacy'' this is).

\subsection{Overview of Techniques}
\label{sec:technique}
We now outline the main ideas behind our results, and outline some possible approaches to improve upon them.

\paragraph{Local privacy.} The algorithm underlying \cref{theo:ldp:private:coin:intro} relies on \emph{Hadamard response}, which was proposed in~\cite{AcharyaSZ19} as a communication-efficient mechanism under local differential privacy constraints. As shown in~\cite{AcharyaCFST21}, when the privacy parameters of two groups of users are identical, this mechanism allows reducing the original closeness testing problem to testing whether the $\lp[2]$-distance between the mean vector of two product-Bernoulli distributions (product distributions over $\bool^\dims$) is $0$ or larger than some parameter. Any sample-optimal $\lp[2]$-closeness testing algorithm for product-Bernoulli distributions can be used to achieve this task efficiently. Here, noise is added to samples from these two product-Bernoulli distributions to preserve privacy. When the privacy parameter becomes smaller (more privacy), the two product-Bernoulli distributions are perturbed by more noise, \ie each attribute in the mean vector will be closer to $\frac{1}{2}$. However, when the privacy parameters are heterogeneous this reduction does not go through as is, since even in the case when $\p=\q$ the transformation will lead to two product-Bernoulli distributions with distinct mean vectors. In particular,  previous $\lp[2]$-closeness (of the mean vectors) testing algorithm for product distributions can no longer be used as a blackbox to achieve the task.

To adapt it to the heterogeneous privacy case, we considered a test statistic for testing closeness of two product-Bernoulli distributions with different levels of noise. A side product is a new sample-optimal $\lp[2]$-closeness tester for product distributions, which simplifies previous algorithms (which were designed for a more general problem, either testing closeness of the product distributions in total variation distance instead of $\lp[2]$ distance of the mean vectors, or without the independence assumption in the soundness case).

\cref{theo:ldp:public:coin:intro} then follows from combining~\cref{theo:ldp:private:coin:intro} and the \emph{domain compression} primitive, which was proposed in~\cite{AcharyaCHST20} as a general technique to derive public-coin schemes from private-coin ones: at a high level, the idea is for the users to leverage public randomness in order to jointly hash the domain of size $\ab$ into $\ab'\ll \ab$ parts, and to consider the induced probability distributions on these $\ab'$ parts. This was shown to preserve the total variation distance between probability distributions up to a shrinking factor of $\sqrt{\ab'/\ab}$, effectively ``replacing'' $\dst$ by $\dst' \asymp \dst\sqrt{\ab'/\ab}$. Selecting the optimal value of $\ab'$ to minimize the resulting sample complexity when applying the private-coin algorithm with the new parameters $\ab'$ and $\dst'$ (in this case, $\ab'=2$) leads to the public-coin algorithm.

\paragraph{Shuffle privacy.} Turning to shuffle privacy, the algorithm behind \cref{theo:shuffle:private:coin:intro} starts with the following observation, which was somewhat implicit in~\cite{CanonneL22} in the context of uniformity testing: if $\ns$ users get each a sample from some distribution $\p$ over $[\ab]$ and use the distributed Poisson mechanism with parameter $\frac{\mu}{\ns} = \bigO{{1}/{(\ns\priv^2)}}$ to ``privately report'' their data in the shuffle model, then the central server gets access to $N\eqdef \ns+\ab\mu$ \iid samples from a new mixture distribution
\[
    \p' \eqdef (1-\gamma)\cdot\p + \gamma\cdot \uniform_{\ab}
\]
where $\gamma \eqdef \ab\mu/N$ and $\uniform_{\ab}$ is the uniform distribution on $[\ab]$. 
(This is not totally accurate as stated, but becomes true if we replace ``$\ns$ users'' and ``$N$ samples'' by ``$\poisson{\ns}$ users'' and $\poisson{N}$ samples.'') In our case, this means that we can obtain $N_1$ ``$\priv_1$-private'' samples from $\p'$ and $N_2$ ``$\priv_2$-private'' samples from $\q'$, where
\begin{align*}
    \p' &\eqdef (1-\gamma_1)\cdot\p + \gamma_1\cdot \uniform_{\ab} \\
    \q' &\eqdef (1-\gamma_2)\cdot\q + \gamma_2\cdot \uniform_{\ab}
\end{align*}
where $\gamma_1 \eqdef \ab\mu_1/N_1=\bigO{\frac{\ab}{\ns_1N_1\priv_1^2}}$ and $\gamma_2 \eqdef \ab\mu_2/N_2=\bigO{\frac{\ab}{\ns_2N_2\priv_2^2}}$ (ignoring again, for the sake of this discussion, the dependence on the second privacy parameter $\privdelta$). Now, a natural requirement is to ask that $\p'=\q'$ whenever $\p=\q$, so that our original private closeness testing task on $\p,\q$ reduces to a new one, \emph{non-private}, on $\p',\q'$, for which we can leverage the existing results on closeness testing. 

Doing so boils down to enforcing $\gamma_1=\gamma_2$, which in turn leads to the requirement $\ns_2 = \bigTheta{({\priv_1}/{\priv_2})^2\ns_1}$; this also gives the ``new'' distance parameter $\dst' \eqdef (1-\gamma_1)\dst=(1-\gamma_2)\dst$ for resulting closeness testing problem on $\p',\q'$. All that remains is to now invoke an existing and optimal (non-private) closeness testing algorithm with unequal numbers of samples (since $N_1\neq N_2$), \eg that of~\cite{DK:16}, and derive the resulting conditions on $N_1,N_2$ (and thus on $\ns_1,\ns_2$) this yields to establish~\cref{theo:shuffle:private:coin:intro}.

As in the locally private case, the public-coin case (\cref{theo:shuffle:public:coin:intro}) then follows \textit{via} the domain compression technique, by selecting the optimal number of parts $\ab'\eqdef \ab'(\ab,\dst,\priv_1,\priv_2)$ to hash the domain into, to minimize the resulting sample complexity when plugging $\ab'$ and $\dst'\asymp\dst\sqrt{\ab'/\ab}$ into~\cref{theo:shuffle:private:coin:intro}, subject to $2\leq \ab'\leq \ab$.\medskip

\noindent\emph{What about amplification by shuffling?} We note that a natural idea to obtain a (different) private-coin shuffle private algorithm (and, via domain compression, a corresponding-public-coin one as well) would be to start from a locally private algorithm under heterogeneous privacy constraints and apply the amplification by shuffling result of~\cite{FeldmanMT21}. This idea was used in~\cite{CanonneL22} for \emph{identity} testing (under homogeneous privacy constraints), \ie testing whether an unknown distribution is equal to a fully known reference one. However, this approach comes with two conditions on the LDP protocol one starts with: (1)~all users of the same group must use the same local randomizer, and (2)~the protocol needs to work reasonably well even for large privacy parameters $\priv_1,\priv_2\gg 1$, \ie in the low-privacy regime. Unfortunately, the LDP protocol behind~\cref{theo:ldp:private:coin:intro} satisfies (1) but not~(2), and thus amplification by shuffling would not lead to a shuffle private algorithm with good enough sample complexity; and the LDP identity testing algorithm of~\cite{CanonneL22} does not seem to generalize to closeness testing, let alone closeness testing under heterogeneous privacy constraints. We believe that obtaining a sample-optimal LDP algorithm under heterogeneous privacy constraints satisfying (1) and (2) could lead to an improvement over~\cref{theo:shuffle:private:coin:intro,theo:shuffle:public:coin:intro} (in terms of $\ns_2$), and leave this as an interesting future direction.

\paragraph{Central privacy.} Finally,~\cref{theo:central:intro} follows from combining two ingredients: the first is the (sample-optimal) closeness testing algorithm of~\cite{Zhang21} for the central model of differential privacy for the equal-privacy case, which has sample complexity
\[
\bigO{\max\Paren{\frac{\ab^{1/2}}{\dst^2}, \frac{\ab^{1/2}}{\priv^{1/2}\dst}, \frac{\ab^{2/3}}{\dst^{4/3}}, \frac{\ab^{1/3}}{\priv^{2/3}\dst^{4/3}}, \frac{1}{\priv\dst}}}\,.
\]
This algorithm relies on adding suitably calibrated noise to the non-private, but low-sensitivity closeness testing algorithm of~\cite{DiakonikolasGKPP21}; this low sensitivity (\ie robust to changing any single sample) is crucial when bounding the noise required to make the algorithm differentially private, as adding Laplace noise with parameter $O(1/\priv)$ \emph{independent of $\ns,\ab,\alpha$} then suffices. 

A natural idea to adapt this to the heterogeneous privacy case would be, as for our shuffle privacy protocol, to first extend this algorithm to the ``uneven-sample case'' (i.e., $\ns_2 \geq \ns_1$) and then introduce two different levels of noise, balancing $\ns_1,\ns_2$ accordingly in order to achieve privacy with parameters $\priv_1,\priv_2$. Unfortunately, as we discuss further in~\cref{ssec:central:dp}, this approach turns out to be much trickier than expected, as the main algorithm of~\cite{DiakonikolasGKPP21} does not appear to easily generalize to the  $\ns_1 \neq \ns_2$ case. Worse, any attempt to do so (as done in~\cite{DiakonikolasGKPP21} for this particular algorithm, via the so-called ``flattening'' technique, or using the other non-private, uneven-sample closeness testing algorithms available in the literature) results in algorithms with \emph{much} higher sensitivity, requiring too large a level of random noise in the ``privatization'' step and leading to vacuous sample complexity bounds.

Instead, we take recourse to a much simpler technique, \emph{privacy amplification by subsampling}~\cite{LiQS12} (as used for a similar goal in~\cite{JorgensenYC15}). The idea is to use the above algorithm of~\cite{Zhang21} with privacy parameter $\priv_1$, requiring $\ns_1$ samples (where $\ns_1$ is given by the above equation) from both groups. Now, this achieves $\priv_1$-privacy for the first group; as for the second, which requires a better privacy guarantee, we start with a larger group of $\ns_2$ users and subsample by picking uniformly at random a subgroup of $\ns_1$ users. By a standard argument, this improves the privacy guarantee from $\priv_1$ to $\frac{\ns_1}{\ns_2}\priv_1=\priv_2$, as desired.
\subsection{Related Work}

Uniformity, identity, and closeness testing are three of flagship (and related) questions in distribution testing, with a long history in classical statistics, where the analysis is under an asymptotic regime, as the number of samples goes to infinity. In contrast, computer scientists often study these problems under the framework of property testing, where one wishes to achieve certain accuracy with a limited number of samples (\ie a particular focus on the ``finite-sample'' regime). For this regime, goodness-of-fit testing without privacy constraints has been extensively studied, with sample complexity bounds summarized in~\cref{tab:testing}\footnote{It is well known that uniformity testing and identity testing share the same upper bound and lower bound of sample complexity, as these two problems can be reduced to each other. Thus, we only include 'identity testing' in this table. }. We refer the readers to~\cite{Canonne:first:survey,BW:17,CanonneTopicsDT22} for surveys of the area.

\begin{table}[ht]
	\caption{Sample complexity bounds of goodness-of-fit testing without privacy constraints}
	\centering
	\begin{tabular}{lll}
		\toprule
		\cmidrule(r){1-2}
		Testing     & Upper bound     & Lower bound \\
		\midrule
		Identity testing  & $\bigO{ \frac{\ab^{1/2}}{\dst^{2}}}$   & tight     \\
		Closeness testing & $\bigO{ \frac{\ab^{1/2}}{\dst^{2}} + \frac{\ab^{2/3}}{\dst^{4/3}}}$   &  tight\\
		\bottomrule
	\end{tabular}
	\label{tab:testing}
\end{table}
A significant body of work has considered these questions under (homogeneous) DP constraints: we list previous results in~\cref{tab:testing:DP}. Most relevant to our work, ~\cite{Zhang21} proposed a sample-optimal closeness-testing algorithm under central differential privacy by leveraging the (non-private) test statistic of~\cite{DiakonikolasGKPP21}.~\cite{AcharyaCFST21} gave sample-optimal identity testing algorithms under local differential privacy for both private-coin and public-coin settings. For shuffle differential privacy,~\cite{BCJM:20,CanonneL22} gave identity testing algorithms with the same complexity bounds on the required number of samples (approximate DP), and~\cite{CheuY21} later provided analogous bounds in the pure privacy setting; however, the exact sample complexity of the question for homogeneous shuffle privacy (either pure or approximate) remains open.

\begin{table}[ht]
\footnotesize
	\caption{Sample complexity bounds for goodness-of-fit testing under homogeneous differential privacy. (The shuffle privacy bounds from~\cite{CheuY21}, marked with ``*'', hold for \emph{pure} privacy, \ie without the logarithmic factor in $1/\privdelta$).}
	\centering
	\begin{tabular}{llll}
		\toprule
		  & \bf Testing question    &\bf  Upper bound     &\bf  Lower bound \\

		\midrule
		\multirow{2}{*}{\bf Local} & Identity, private-coin     & $\bigO{\frac{\ab^{3/2}}{\dst^{2} \priv^2}} $ ~\cite{AcharyaCFST21} & tight ~\cite{AcharyaCT:IT1}     \\
    &Identity, public-coin &$\bigO{\frac{\ab}{\dst^{2} \priv^2}} $ ~\cite{AcharyaCFST21} & tight ~\cite{AcharyaCT:IT1,AJM:20} \\
        \cline{2-4}
		& Closeness, private-coin     & No result   &  $\bigOmega{\frac{\ab^{3/2}}{\dst^{2} \priv^2}} $ ~\cite{AcharyaCT:IT1} \\
        & Closeness, public-coin     & No result   &  $\bigOmega{\frac{\ab}{\dst^{2} \priv^2}} $ ~\cite{AcharyaCT:IT1,AJM:20} \\

        \midrule
		\multirow{2}{*}{\bf Shuffle} & Identity, private-coin     &  $\bigO{\frac{\ab^{3/4}}{\dst \priv}\sqrt{\log(\frac{1}{\privdelta})}+\frac{\sqrt{\ab}}{\dst^2}}$ ~\cite{BCJM:20,CanonneL22},~\cite{CheuY21}*   &   $\bigOmega{\frac{\ab^{2/3}}{\dst^{4/3}\priv^{2/3}}+\frac{\sqrt{\ab}}{\dst^2} + \frac{1}{\dst \priv}}$~\cite{BCJM:20}      \\ 
        &Identity, public-coin & $\bigO{\frac{\ab^{2/3}}{\dst^{4/3}\priv^{2/3}}\log^{1/3}\frac{1}{\privdelta}+\frac{\sqrt{\ab}}{\dst\priv}\log^{1/2}\frac{1}{\privdelta}+\frac{\sqrt{\ab}}{\dst^2}}$~\cite{CanonneL22},~\cite{CheuY21}* & $\bigOmega{\frac{\ab^{2/3}}{\dst^{4/3}\priv^{2/3}}+\frac{\sqrt{\ab}}{\dst^2} + \frac{1}{\dst \priv}}$~\cite{BCJM:20} \\
        \cline{2-4}
        
		&Closeness (either)    & No result   &  $\bigOmega{\frac{\ab^{2/3}}{\dst^{4/3}\priv^{2/3}}+\frac{\sqrt{\ab}}{\dst^2} + \frac{1}{\dst \priv}}$~\cite{BCJM:20}  \\
  
		\midrule
		\multirow{2}{*}{\bf Central} & Identity  & $\bigO{ \frac{\ab^{1/2}}{\dst^{2}} + \frac{\ab^{1/2}}{\dst \priv^{1/2}} +
		\frac{\ab^{1/3}}{\dst^{4/3} \priv^{2/3}} + \frac{1}{\dst \priv}}$ ~\cite{ASZ:18:DP}  & tight ~\cite{ASZ:18:DP}     \\
		& Closeness & $\bigO{ \frac{\ab^{1/2}}{\dst^{2}} + \frac{\ab^{2/3}}{\dst^{4/3}} + \frac{\ab^{1/2}}{\dst \priv^{1/2}} +
		\frac{\ab^{1/3}}{\dst^{4/3} \priv^{2/3}} + \frac{1}{\dst \priv}}$ ~\cite{Zhang21}  &  tight ~\cite{Zhang21} \\
  
		\bottomrule
	\end{tabular}
    \label{tab:testing:DP}
\end{table}
\subsection{Organization of the Paper}
In~\cref{sec:preliminaries}, we give the formal definition of our problem and introduce some necessary tools used in this paper. In~\cref{sec:algorithm}, we present our algorithms and the related proofs, and give a discussion on closeness testing under homogeneous central differential privacy constraint. In~\cref{sec:future:work}, we discuss future work.

\section{Model and Preliminaries}
    \label{sec:preliminaries}
\subsection{Closeness Testing}
Given sample access to two unknown distributions, closeness testing asks whether these distributions are the same, or differ significantly in terms of statistical distance:\footnote{The statistical (total variation) distance between two probability distributions $\p,\q$ over the same domain $\domain$ is defined as $\totalvardist{\p}{\q} = \sup_{S\subseteq \domain} (\p(S)-\q(S) = \frac{1}{2} \normone{\p-\q}$, where the supremum is taken over all measurable subsets.}
\begin{definition}[Closeness testing]
Let $\p,\q$ be two unknown distributions with domain $[\ab]$. A closeness testing algorithm with sample complexity $\ns$ takes inputs $\dst \in (0, 1]$, a set of $\ns$ i.i.d. samples from $\p$ and  a set of $\ns$ i.i.d. samples from $\p$ and outputs either \accept or \reject such that the following holds:
\begin{itemize}
    \item If $\p = \q$, then the algorithm outputs \accept with probability at least $\frac{2}{3}$;
    \item If $\totalvardist{\p}{\q} > \dst$, then the algorithm outputs \reject with probability at least $\frac{2}{3}$.
\end{itemize}
\end{definition}
In this definition, $\frac{2}{3}$ is just some arbitrary number picked between $\frac{1}{2}$ and $1$. By a standard amplification argument any high probability $1-\errprob$ can be achieved by repeating the test independently $O(\log(1/\errprob)$ times and using a majority rule.

\subsection{Differential Privacy}

We now recall the relevant concepts we will extensively use, starting with the definition of differential privacy.
\begin{definition}
[Differential privacy~\cite{DMNS:06}]  A randomized algorithm $\mathcal{M}\colon\mathcal{X}^n \to \mathcal{R}^{k}$ is said to be \emph{$(\priv, \privdelta)$-differentially private} if for all measurable $S \subseteq \mathcal{R}^{k}$ and all neighboring datasets $x,y \in \mathcal{X}$:
\[
    \bPr{\mathcal{M}(x) \in S} \leq e^{\priv} \bPr{\mathcal{M}(y) \in S} + \privdelta\,,
\]
where two datasets $x,y$ are said to be neighboring if $\dist{x}{y} \leq 1$ (\ie for $\dist{}{}$ being the Hamming distance, if they only differ in (at most) one entry).
\end{definition}
A key property of differential privacy is \emph{immunity of post-processing}.
\begin{lemma}[Immunity of post-processing]
    \label{theo:immunity:post-processing}
    Let $f \colon \mathcal{X}^n \rightarrow \mathcal{R}^k$ be a mapping which is $(\priv,\privdelta)$-differentially private. Let $g \colon \mathcal{R}^k \rightarrow \mathcal{R}^{k'}$ be any arbitrary random mapping. Then the mapping $g \circ f$ is still $(\priv,\privdelta)$-differentially private.
\end{lemma}

In many cases, this randomized mapping $\mathcal{M}$ is obtained by adding random noise to some function $f(x)$ of the data, where $f$ is the (non-private) function to be computed (such mechanisms are called \emph{additive noise mechanisms}). The amount of noise to be added to $f(x)$ then needs to be tailored to the specific properties of $f$: in particular, such an important property is its $\lp[1]$-sensitivity. 
\begin{definition}[Sensitivity]
    The $\lp[1]$-sensitivity of a function $f \colon \mathcal{N}^{|\mathcal{X}|} \rightarrow \mathcal{R}^{k}$ is  
    \[
        \triangle  = \max_{x,y \in \mathcal{N}^{|\mathcal{X}|},\, \dist{x}{y} \leq 1} \norm{f(x) - f(y)}_1\,.
    \]
    We say the function $f$ is \emph{$\triangle$-sensitive}. 
\end{definition}

One mechanism for adding randomness using $\lp[1]$-sensivity as a parameter is the Poisson mechanism. We use this mechanism in the shuffle model.
\begin{lemma}[Poisson mechanism~\cite{Ghazi0MP20}]
    \label{lemma:poisson_mech} 
    Let $f \colon \mathcal{X}^n \rightarrow \mathcal{Z}$ be a $\triangle$-sensitive function. For any $\priv > 0, \privdelta \in (0,1)$ and $\lambda \geq \frac{16\log (10/\privdelta)}{(1-e^{-\priv/\triangle})^2} + \frac{2\triangle}{1-e^{-\priv/\triangle}}$, the randomized function $\mathcal{A}(x) = f(x) + Y$ where $x \in \mathcal{X}^n, Y \sim \poisson{\lambda}$ is $(\priv,\privdelta)-$differentially private in the central and shuffle model.
\end{lemma}

A standard technique to amplify a differential-private algorithm is amplification by subsampling. The formal statement is as follows.
\begin{lemma}[{Amplification by subsampling (see, \eg,~\cite{LiQS12}, or\cite[Theorem~9]{BalleBG18})}]
    \label{theo:subsampling}
    Let $\mathcal{A}: \mathcal{X}^{n_1} \rightarrow \mathcal{R}^k$ be a $(\priv_1,\privdelta_1)$-differentially private mapping and $\mathcal{H}\colon \mathcal{X}^{n_2} \rightarrow \mathcal{X}^{n_1}$ be a randomized mapping which uniformly chooses $n_1$ elements from $n_2$ elements. Then, the composition of this two mappings $\mathcal{A} \circ \mathcal{H}\colon\mathcal{X}^{n_2} \rightarrow \mathcal{R}^k$ is $(\priv_2,\privdelta_2)$-differentially private where $\priv_2 = \ln (1+ \frac{n_1}{n_2}(e^{\priv_1}-1))$ and $\privdelta_2 = \frac{n_1}{n_2} \privdelta_1$. 
    In particular, if $\priv_1 \leq 1$ then $\priv_2 = \bigO{\frac{n_1}{n_2} \priv_1}$.
\end{lemma}
We conclude this by recalling the privacy of Randomized Response~\cite{Warner:65}, one of the standard mechanisms:
\begin{fact}[Randomized Response]
    \label{fact:RR}
    Let $\mathcal{M}_f$ be the mechanism which takes one bit as input and flips it with probability $\frac{1}{e^{\priv}+1}$. Then, $\mathcal{M}$ is $(\priv,0)-$differentially private.
    \label{fact:flip_1_bit}
\end{fact}

\subsection{Central, local and shuffle models}
In the central model of differential privacy, there is a trusted data curator who holds all the original data and guarantees its output is differentially private. A stricter model is the local model of differential privacy, where the data curator is untrusted and only receives noisy data. It is well known that the local model provides a stronger privacy guarantee, but often at the cost of utility (that is, usually leads to worse accuracy).

A third model, ``between'' the local and central models is the shuffle model. While a central data curator is still not trusted in this model, we allow a ``shuffler'' to receive messages from the users and anonymize them by applying a uniformly random permutation . The permutation is typically implemented using cryptographic primitives such as secure multi-party computation. One advantage of shuffling is that while these cryptographic primitives are time-intensive, implementing a shuffler is simple enough and is not too time-consuming. (Instead, a fully trusted algorithm using cryptography would be too computationally intensive.)

\subsection{Central, Shuffle and Local Models for Heterogeneous Privacy and Data}
    \label{ssec-models}

    Since testing involves two different distributions, it is natural to consider the case where users from different distributions have different concerns of privacy. Specifically, we want to make sure our algorithm is $(\priv_1,\privdelta_1)$-differentially private for samples from $\p$ and $(\priv_2,\privdelta_2)$-differentially private for samples from $\q$. We introduce the corresponding definition of the testingtask below.

\begin{definition}[Closeness testing under heterogeneous local differential privacy constraints]
Let $\p,\q$ be two unknown distributions with domain $[\ab]$. A \emph{closeness testing algorithm under heterogeneous local differential privacy constraints} consists of the following:
\begin{itemize}
    \item $\ns_1$ randomizers $\mathcal{R}_1, ..., \mathcal{R}_{\ns_1}: \mathcal{X} \times \{ 0,1 \}^r \rightarrow \mathcal{Y}$ mapping a sample drawn from $\p$ and a public randomness bit of length  $r$ to a privatised output.
    
    \item $\ns_2$ randomizers $\mathcal{R'}_1, ..., \mathcal{R'}_{\ns_2}: \mathcal{X'} \times \{ 0,1 \}^r \rightarrow \mathcal{Y'}$ mapping a sample drawn from $\q$ and a public randomness bit of length  $r$ to a privatised output.
    
    \item an analyser $\mathcal{A}: \mathcal{Y}^{\ns_1} \times \mathcal{Y}^{\ns_2} \times \{ 0,1 \}^r \rightarrow \mathcal{Z}$ mapping all privatised message and the public randomness bit to the result of analysis either \accept or \reject such that the following holds
    \begin{itemize}
    \item If $\p = \q$, then $\mathcal{A}$ outputs \accept with probability at least $\frac{2}{3}$
    \item If $\totalvardist{\p}{\q} > \dst$, then $\mathcal{A}$ outputs \reject with probability at least $\frac{2}{3}$
    \end{itemize}
\end{itemize}

When $r = 0$ (no public randomness), the testing algorithm is said to belong to the \emph{private-coin local differential-private model}. Otherwise, it belongs to the \emph{public-coin local differential-private model}.

When the output $\mathcal{R}(\mathcal{X})$ of each randomizer $\mathcal{R}$ is $(\priv_1, \privdelta_1)$-differentially private w.r.t. $\mathcal{X}$ and the output $\mathcal{R'}(\mathcal{X'})$ of each randomizer $\mathcal{R'}$ is $(\priv_2, \privdelta_2)$-differentially private w.r.t. $\mathcal{X'}$, $\mathcal{P}$ is said to be \emph{$((\priv_1, \privdelta_1),(\priv_2, \privdelta_2))$-heterogeneously locally differentially private}. For simplicity, $\mathcal{P}$ is said to be \emph{$(\priv_1, \priv_2)$-heterogeneously locally differentially private} when $\privdelta_1,\privdelta_2 = 0$.

\end{definition}

    \begin{figure}[H]\centering
    \scalebox{0.7}{
    \begin{tikzpicture}[->,>=stealth',shorten >=1pt,auto,node distance=20mm, semithick,scale=1.1, every node/.style={transform shape}]
  \node[circle,draw,minimum size=13mm] (A) {$X_1$};
  \node (B) [right of=A] {$\dots$};
  \node[circle,draw,minimum size=13mm] (BB) [right of=B] {$X_{\ns_1}$};
  \node (C) [right of=BB] {};
  \node[circle,draw,minimum size=13mm] (DD) [right of=C] {$X'_1$};
  \node (D) [right of=DD] {$\dots$};
  \node[circle,draw,minimum size=13mm] (E) [right of=D] {$X'_{\ns_2}$};
  
  \node[rectangle,draw,minimum width=13mm,minimum height=7mm] (WA) [below of=A] {$R_1^{\vphantom{'}}$};
  \node (WB) [below of=B] {$\dots$};
  \node[rectangle,draw,minimum width=13mm,minimum height=7mm] (WBB) [below of=BB] {$R_{\ns_1}$};
  
  \node (WC) [below of=C] {};
 
  \node[rectangle,draw,minimum width=13mm,minimum height=7mm] (WDD) [below of=DD] {$R'_1$};
  \node  (WD) [below of=D] {$\dots$};
  \node[rectangle,draw,minimum width=13mm,minimum height=7mm] (WE) [below of=E] {$R'_{\ns_2}$};
  
  \node (P) [above of=B] {$\mathbf{p}$};
  \node (Q) [above of=D] {$\mathbf{q}$};
  \node[rectangle,draw, minimum size=10mm] (R) [below of=WC] {Analyzer};
  \node (out) [below of=R,node distance=13mm] {output $\in\{\accept,\reject\}$};
  
  \node[draw,dashed,fit=(A) (E) (WA) (WB) (WBB) (WC) (WDD) (WD) (WE)] {trusted};

  \draw[->] (P) edge[densely dashed,bend right=10] (A)(A) edge (WA)(WA) edge[bend right=10] (R);
  
  \draw[->] (P) edge[densely dashed] (BB)(BB) edge (WBB)(WBB)  edge (R);
  \draw[->] (Q) edge[densely dashed] (DD)(DD) edge (WDD)(WDD) edge (R);

  \draw[->] (Q) edge[densely dashed,bend left=10] (E)(E) edge (WE)(WE) edge[bend left=10] (R);
  \draw[->] (R) edge (out);
\end{tikzpicture}
    }
    \caption{Closeness testing under heterogeneous \emph{local} differential privacy constraints. The $R_i$'s (resp. $R'_i$'s) are the local randomizers used by the users to privatize their data prior to sending it to the analyzer.}
    \end{figure}
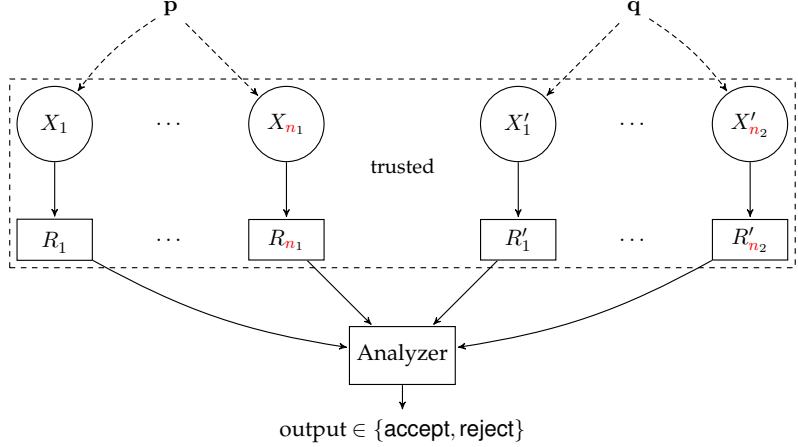

We now provide the analogous definitions for shuffle privacy. As discussed in the introduction, this definition is not entirely straightforward, as the very definition of the closeness testing problem requires the ability to distinguish between two groups of users -- those with inputs from $\p$, and those with inputs from $\q$. This goes against the objective of shuffling, and motivates the introduction of two distinct shufflers: one for each group.
\begin{definition}[Closeness testing under heterogeneous shuffle differential privacy constraints]
Let $\p,\q$ be two unknown distributions with domain $[\ab]$. A closeness testing algorithm under heterogeneous shuffle differential privacy constraints consists of the following:
\begin{itemize}
    \item $\ns_1$ randomizers $\mathcal{R}_1, ..., \mathcal{R}_{\ns_1}: \mathcal{X} \times \{ 0,1 \}^r \rightarrow \mathcal{Y}$ mapping a sample drawn from $\p$ and a public randomness bit of length  $r$ to a privatised output.
    
    \item $\ns_2$ randomizers $\mathcal{R'}_1, ..., \mathcal{R'}_{\ns_2}: \mathcal{X} \times \{ 0,1 \}^r \rightarrow \mathcal{Y}'$ mapping a sample drawn from $\q$ and a public randomness bit of length $r$ to a privatised output.
    
    \item  A shuffler $\mathcal{S}_1 : \mathcal{Y} \rightarrow \mathcal{Y}^*$ that concatenates message vectors and then applies a uniformly random permutation to the messages.
    
    \item  A shuffler $\mathcal{S}_2 : \mathcal{Y}' \rightarrow \mathcal{Y}^{'*} $ that concatenates message vectors and then applies a uniformly random permutation to the messages.

    \item an analyser $\mathcal{A}: \mathcal{Y}^* \times \mathcal{Y}^{'*} \times \{ 0,1 \}^r \rightarrow \mathcal{Z}$ mapping all privatised message and the public randomness bit to the result of analysis  either \accept or \reject such that the following holds
    \begin{itemize}
    \item If $\p = \q$, then $\mathcal{A}$ outputs \accept with probability at least $\frac{2}{3}$
    \item If $\totalvardist{\p}{\q} > \dst$, then $\mathcal{A}$ outputs \reject with probability at least $\frac{2}{3}$
    \end{itemize}
\end{itemize}

When the output of the shuffler $\mathcal{S}_1$ is $(\priv_1, \privdelta_1)$-differentially private w.r.t. $\mathcal{X}$ and the output of the shuffler $\mathcal{S}_2$ is $(\priv_2, \privdelta_2)$-differentially private w.r.t. $\mathcal{X'}$, $\mathcal{P}$ is said to be \emph{$((\priv_1, \privdelta_1),(\priv_2, \privdelta_2))$-heterogeneously shuffle differentially private}. For simplicity, $\mathcal{P}$ is said to be \emph{$(\priv_1, \priv_2)$-heterogeneously shuffle differentially private} when $\privdelta_1,\privdelta_2 = 0$.
\end{definition}
    
    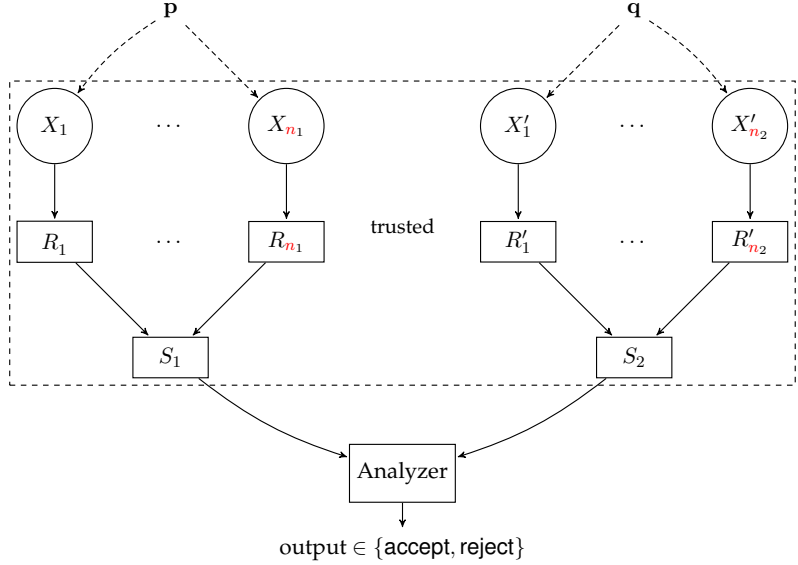
\begin{figure}[H]\centering
    \scalebox{0.7}{
    \begin{tikzpicture}[->,>=stealth',shorten >=1pt,auto,node distance=20mm, semithick,scale=1.1, every node/.style={transform shape}]
  \node[circle,draw,minimum size=13mm] (A) {$X_1$};
  \node (B) [right of=A] {$\dots$};
  \node[circle,draw,minimum size=13mm] (BB) [right of=B] {$X_{\ns_1}$};
  \node (C) [right of=BB] {};
  \node[circle,draw,minimum size=13mm] (DD) [right of=C] {$X'_1$};
  \node (D) [right of=DD] {$\dots$};
  \node[circle,draw,minimum size=13mm] (E) [right of=D] {$X'_{\ns_2}$};
  
  \node[rectangle,draw,minimum width=13mm,minimum height=7mm] (WA) [below of=A] {$R_1^{\vphantom{'}}$};
  \node (WB) [below of=B] {$\dots$};
  \node[rectangle,draw,minimum width=13mm,minimum height=7mm] (WBB) [below of=BB] {$R_{\ns_1}$};
  
  \node (WC) [below of=C] {};
 
  \node[rectangle,draw,minimum width=13mm,minimum height=7mm] (WDD) [below of=DD] {$R'_1$};
  \node  (WD) [below of=D] {$\dots$};
  \node[rectangle,draw,minimum width=13mm,minimum height=7mm] (WE) [below of=E] {$R'_{\ns_2}$};
  
  \node[rectangle,draw,minimum width=13mm,minimum height=7mm] (SD) [below of=WD] {$S_2$};
  \node[rectangle,draw,minimum width=13mm,minimum height=7mm] (SB) [below of=WB] {$S_1$};
  
  \node (P) [above of=B] {$\mathbf{p}$};
  \node (Q) [above of=D] {$\mathbf{q}$};
  \node (R) [below of=WC] {};
  \node[rectangle,draw, minimum size=10mm] (RR) [below of=R] {Analyzer};
  \node (out) [below of=RR,node distance=13mm] {output $\in\{\accept,\reject\}$};
  
  \node[draw,dashed,fit= (A) (E) (WA) (WB) (WBB) (WC) (WDD) (WD) (WE) (SB) (SD)] {trusted};

  \draw[->] (P) edge[densely dashed,bend right=10] (A)(A) edge (WA)(WA) edge (SB)(SB) edge[bend right=10] (RR);
  
  \draw[->] (P) edge[densely dashed] (BB)(BB) edge (WBB)(WBB)  edge (SB);
  \draw[->] (Q) edge[densely dashed] (DD)(DD) edge (WDD)(WDD) edge (SD);

  \draw[->] (Q) edge[densely dashed,bend left=10] (E)(E) edge (WE)(WE) edge (SD)(SD) edge[bend left=10] (RR);
  \draw[->] (RR) edge (out);
\end{tikzpicture}
    }
    \caption{Closeness testing under heterogeneous \emph{shuffle} differential privacy constraints. Here, $S_1$ and $S_2$ are the shufflers for the two groups, and the $R_i$'s (resp. $R'_i$'s) are the randomizers used by the users, prior to the shuffling.}
    \end{figure}

Finally, we conclude with the definition of the task under (central) differential privacy:
\begin{definition}[Closeness testing under heterogeneous central differential privacy constraints]
Let $\p,\q$ be two unknown distributions with domain $[\ab]$. A closeness testing algorithm under central local differential privacy constraints consists of the following:
\begin{itemize}

    \item an analyser $\mathcal{A}: \mathcal{X} \times \mathcal{X'} \times \{ 0,1 \}^r \rightarrow \mathcal{Z}$ mapping a sample vector $\mathcal{X}$ from $\p$ and a sample vector $\mathcal{X'}$ from $\q$ to the result of analysis either \accept or \reject such that the following holds
    \begin{itemize}
    \item If $\p = \q$, then $\mathcal{A}$ outputs \accept with probability at least $\frac{2}{3}$
    \item If $\totalvardist{\p}{\q} > \dst$, then $\mathcal{A}$ outputs \reject with probability at least $\frac{2}{3}$
    \end{itemize}
\end{itemize}

When the output $\mathcal{A}(\mathcal{X})$ of the analyser $\mathcal{A}$ is $(\priv_1, \privdelta_1)$-differentially private w.r.t. $\mathcal{X}$ and $(\priv_2, \privdelta_2)$-differentially private w.r.t. $\mathcal{X'}$, $\mathcal{P}$ is said to be \emph{$((\priv_1, \privdelta_1),(\priv_2, \privdelta_2))$-heterogeneously centrally differentially private}. For simplicity, $\mathcal{P}$ is said to be \emph{$(\priv_1, \priv_2)$-heterogeneously centrally differentially private} when $\privdelta_1,\privdelta_2 = 0$.

\end{definition}
    
    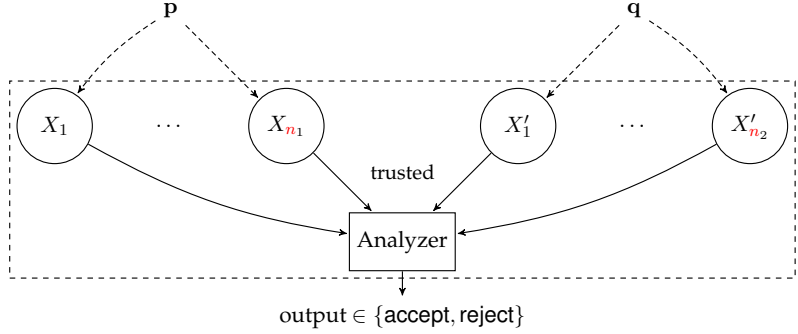
\begin{figure}[H]\centering
    \scalebox{0.7}{
    \begin{tikzpicture}[->,>=stealth',shorten >=1pt,auto,node distance=20mm, semithick,scale=1.1, every node/.style={transform shape}]
  \node[circle,draw,minimum size=13mm] (A) {$X_1$};
  \node (B) [right of=A] {$\dots$};
  \node[circle,draw,minimum size=13mm] (BB) [right of=B] {$X_{\ns_1}$};
  \node (C) [right of=BB] {};
  \node[circle,draw,minimum size=13mm] (DD) [right of=C] {$X'_1$};
  \node (D) [right of=DD] {$\dots$};
  \node[circle,draw,minimum size=13mm] (E) [right of=D] {$X'_{\ns_2}$};

  \node (P) [above of=B] {$\mathbf{p}$};
  \node (Q) [above of=D] {$\mathbf{q}$};
  \node[rectangle,draw, minimum size=10mm] (R) [below of=C] {Analyzer};
  \node (out) [below of=R,node distance=13mm] {output $\in\{\accept,\reject\}$};
  
  \node[draw,dashed,fit=(R) (A) (E)] {trusted};

  \draw[->] (P) edge[densely dashed,bend right=10] (A)(A) edge[bend right=10] (R);
  
  \draw[->] (P) edge[densely dashed] (BB)(BB) edge  (R);
  \draw[->] (Q) edge[densely dashed] (DD)(DD) edge (R);

  \draw[->] (Q) edge[densely dashed,bend left=10] (E)(E)  edge[bend left=10] (R);
  \draw[->] (R) edge (out);
\end{tikzpicture}
    }
    \caption{Closeness testing under heterogeneous \emph{central} differential privacy constraints.}
    \end{figure}

\subsection{Domain compression}

Finally, to obtain public-coin protocols from private-coin ones, we will rely on the following \emph{domain compression} result, a hashing-type technique that allows to trade domain size for distance parameter in distribution testing: 
\begin{lemma}[Domain Compression~\cite{AcharyaCHST20}]
  \label{theo:random:dct:hashing}
  There exist absolute constants $c_1,c_2>0$ such that the following holds. For any $2\leq \ab'\leq \ab$ and any distributions $\p,\q$ over $[\ab]$,
  \[
        \probaDistrOf{\Pi}{ \totalvardist{\p_\Pi}{\q_\Pi} \geq c_1\sqrt{\frac{\ab'}{\ab}}\totalvardist{\p}{\q} } \geq c_2\,,
  \] 
  where $\Pi=(\Pi_1,\dots\Pi_{\ab'})$ is a uniformly random partition of $[\ab]$ in $\ab'$ subsets, and $\p_\Pi$ denotes the probability distribution on $[\ab']$ induced by $\p$ and $\Pi$ via $\p_\Pi(i) = \p(\Pi_i)$.
\end{lemma}

\paragraph{Notation.} Throughout, we write $a_n \gtrsim b_n$ (resp. $a_n \lesssim b_n$) to denote the existence of an absolute constant $C>0$ such that $a_n \leq C\cdot b_n$ (resp. $a_n \geq C\cdot b_n$) for all $n$; and use $a_n \asymp b_n$ when both $a_n \gtrsim b_n$ and $a_n \lesssim b_n$ hold. Besides this, we use the standard $O(\cdot)$, $\Omega(\cdot)$, and $\Theta(\cdot)$ notation. Hereafter, we identify a probability distribution $\p$ over a discrete domain $\domain$ with its probability mass function (pmf), writing $\p(x)$ for $\probaDistrOf{X\sim \p}{X=x}$; and for a subset $S$ of the domain, write $\p(S) = \sum_{x\in S}\p(x)$.
\section{Our Algorithms}
    \label{sec:algorithm}
We now provide the details of our algorithms, starting with those under (heterogeneous) local privacy.
\subsection{Under Local Privacy}
\subsubsection{Private-coin protocol}
Our algorithm for testing closeness under heterogeneous local differential privacy constraints is based on the Hadamard Response mechanism~\cite{AcharyaS19}; we recall one of its key properties below.
\begin{theorem}[\cite{AcharyaS19}]
    \label{theo:ldb_l1_to_l2_mean}
    Let $H^{K}$ be the $K \times K$ Hadamard matrix where $K = 2^{\lfloor \log_2(\ab + 1) \rfloor}$, which is the smallest power of two larger than $\ab$. Let $C_j$ be the locations of $1$s in the $j$th column of $H^{K}$ where $j \in [K]$. For any distribution $\p$, let $\p(C_j)$ be the probability that a sample from $\p$ falls in set $C_j$. Then we have
    \[
        \sum_{j=1}^{\ab}(\p(C_j) - \q(C_j))^2 = \frac{K}{4}\norm{\p-\q}^2_2.
    \]
\end{theorem}

Recall that identity and closeness testing fix a distance $\dst$, and test whether two distributions $\p,\q$ are the same or the total variation distance between two distributions is larger than $\dst$. By using the Cauchy--Schwarz inequality, we have $\sqrt{\ab} \norm{\p-\q}_2  \geq \norm{\p-\q}_1$ for any two distributions $\p,\,\q \in [\ab]$. Thus, if $\totalvardist{\p}{\q} = \frac{1}{2}\norm{\p-\q}_1 > \dst$ then we must have $\frac{4}{K} \sum_{j=1}^{\ab}(\p(C_j) - \q(C_j))^2 = \norm{\p-\q}^2_2 \geq (\frac{1}{\sqrt{\ab}}\norm{\p-\q}_1)^2 > \frac{4}{\ab}\dst^2$. Since $K \geq \ab$, that implies

\begin{equation}
\sum_{j=1}^{\ab}(\p(C_j) - \q(C_j))^2 > \dst^2.
\label{eq:local:hadmard_pq_not_eq}
\end{equation}
Otherwise, if $\p = \q$, we must have
\begin{equation}
        \sum_{j=1}^{\ab}(\p(C_j) - \q(C_j))^2 = \frac{K}{4}\norm{\p-\q}^2_2 = 0.
\end{equation}
Motivated by this observation, the identity testing algorithm in \cite{AcharyaCFST21} divides users from distribution $\p$ into $K$ disjoint groups with equal size. It then assigns the users in the $j$th group to a set $C_j$. Each user generates a 1-bit message, indicating whether the data of the user belongs to $C_j$. To make the output of each user differentially private, each user needs to add some noise to their output. Specifically, each user flips the $1$ bit of message with a certain probability $\frac{1}{e^{\priv}+1}$ and sends it using Randomized Response~\cref{fact:RR}. 

Recall that the message sent by each user is only one bit, and the messages sent in the same group follows the same Bernoulli distribution. For users from distribution $\p$, by taking the message from one user from every set, we can obtain a sample from a product-Bernoulli distribution $P$ with length $K$. Let $\mu(P)$ be the mean of the product-Bernoulli distribution, we have
\begin{equation}
        \mu(P)_j = \frac{e^{\priv}}{e^{\priv}+1}\sum_{x \in C_j}\p(x)   + \frac{1}{e^{\priv}+1} \sum_{x \notin C_j}\p(x)  = \frac{e^{\priv}-1}{e^{\priv}+1}\p(C_j) + \frac{1}{e^{\priv}+1}
        \label{eq:local:distribution_with_noise}
\end{equation}

Since our task is closeness testing instead of identity testing, we also need to perform the same operations for users from $\q$ and obtain samples from another product-Bernoulli distribution $Q$. Similarly, we use $\mu(Q)$ to denote the mean of the product-Bernoulli distribution $Q$.

There is a very intuitive understanding of~\cref{eq:local:distribution_with_noise}. When $\priv \rightarrow \infty$ (no privacy), $\mu(P)_j \rightarrow \p(C_j)$. When $\priv \rightarrow 0$ (no accuracy), $\mu(P)_j$ converges to $\frac{1}{2}$. Moreover, if users from two distributions are using the same parameter $(\priv,0)$ for differential privacy and $\totalvardist{\p}{\q} \geq \dst$, we have $\norm{\mu(P) - \mu(Q)}_2 > \frac{e^{\priv}-1}{e^{\priv}+1} \dst$ by combining~\cref{eq:local:hadmard_pq_not_eq,eq:local:distribution_with_noise}.

That implies we can test whether $\p = \q$ or $d_{TV} (\p,\q) > \dst$ by \emph{non-privately} testing whether $\mu(P)=\mu(Q)$ or $\norm{\mu(P)-\mu(Q)}_2 > \frac{e^{\priv}-1}{e^{\priv}+1} \dst$., using any sample-optimal $\lp[2]$-testing algorithm for testing identity and closeness of product distributions (\eg \cite{CanonneDKS20}). This leads to the optimal sample complexity for identity testing under LDP constraints, as shown in \cite{AcharyaCFST21}.

If we want to use the same privatizing method with heterogeneous constraints, however, we can no longer simply use a $\ell_2$ closeness testing algorithm for product distributions as outlined above: indeed, our differential privacy constraints for the two distributions $\p,\q$ are not the same. By \cref{fact:flip_1_bit} we need to flip the bits of message from wo distributions with probabilities $\frac{1}{e^{\priv_1}+1}$ and $\frac{1}{e^{\priv_1}+1}$ respectively, if we want messages from $\p$ to be $(\priv_1,0)$-differentially private and messages from $\q$ to be $(\priv_2,0)$-differentially private. But that implies $\mu(P)$ and $\mu(Q)$ will not be the same even if $\p = \q$. Thus, we need to have a different algorithm for testing closeness between $\p,\q$ given $\mu(P),\mu(Q)$. To design this algorithm, we first provide an algorithm for closeness testing of product distributions taylored to our purpose, which is simpler than the (more general) algorithm in~\cite{CanonneDKS20}. The idea of this testing algorithm is inspired by the test statistic in~\cite{CanonneCKLW21}, which we can simplify as we do not need, in our case, to deal with arbitrary covariances. While the end guarantees are not new, we provide this slightly simpler algorithm for completeness.

\begin{algorithm}
\caption{Testing closeness of two product distributions $P,\,Q$}\label{alg:product_closeness_no_privacy}
\begin{algorithmic}[1]
\Require Two groups of samples $X^{(1)},...,X^{(\ns)},X'^{(1)},...,X'^{(\ns)}$  from the $d-$dimensional product distribution $P$ and two groups of samples $Y^{(1)},...,Y^{(\ns)},Y'^{(1)},...,Y'^{(\ns)}$  from the $d-$dimensional product distribution $Q$, where $P,\,Q \in \{ -1,\,1 \}^d$ and $\ns = \frac{100\sqrt{d}}{\dst^2} $
\State Calculate $\hat{X}, \hat{X'}, \hat{Y}, \hat{Y'}$, which are mean vectors of $X,X',Y,Y'$ respectively.
\State Define $Z_1 = \langle\, \hat{X} - \hat{Y} , \hat{X’} - \hat{Y'} \rangle$.
\If{$Z_1 \leq \frac{1}{2}\dst^2$}
    \State return accept.
\Else 
    \State return reject.
\EndIf
\end{algorithmic}
\end{algorithm}

\begin{lemma}
    \label{theo:product_closeness_no_privacy}
    Assume we can draw samples from two $d$-dimensional product-Bernoulli distributions $P, Q \in \{ 0,\,1 \}^d$. Given a distance parameter $\dst > 0$,~\cref{alg:product_closeness_no_privacy} is a sample-optimal algorithm which distinguishes between $P = Q$ and $\norm{\mu(P) - \mu(Q)}_2 > \dst$ with probability at least $\frac{2}{3}$ using $\bigO{\frac{\sqrt{d}}{\dst^2}}$ samples.
\end{lemma}
\noindent The proof can be found in~\cref{app:proof:product_closeness_no_privacy}.\smallskip

We then show that~\cref{alg:product_closeness_no_privacy}, with some modification, can deal with two groups of samples under different differential privacy constraints.

\begin{theorem}
    \label{theo:l2_mean_flip}
    Assume we can draw samples with noise from two $d$-dimensional product distributions $P, Q \in \{ 0,\,1 \}^d$ , where each coordinate in samples from $P, Q$ is flipped with probability $\frac{1}{e^{\priv_1}+1},\,\frac{1}{e^{\priv_2}+1}$ respectively. Given privacy parameters $\priv_1, \priv_2 \in (0,1]$ and a distance parameter $\dst$, there exist an sample-optimal algorithm which uses $\frac{\sqrt{d}}{\dst^2 \priv_1^2}$ samples from $P$ and $\frac{\sqrt{d}}{\dst^2 \priv_2^2}$ samples from $Q$, and distinguish between $\norm{\mu(P)-\mu(Q)}_2 = 0$ and $\norm{\mu(P) - \mu(Q)}_2 > \dst$ with high probability.
\end{theorem}
\noindent In the interest of space and clarity of exposition, the proof of this theorem is deferred to~\cref{app:proof:l2_mean_flip}.

Finally, we can claim that we have an sample-optimal algorithm for testing closeness of two distributions under heterogeneous local privacy constraints.

\begin{theorem}
    \label{theo:ldp_ns}
    For every $\ab \ge 0$ and $\priv_1, \priv_2 \in (0,1]$, there exist a private-coin protocol for $(\ab,\dst)$-closeness testing between two unknown distributions $\p,\,\q$ using $\bigO{\frac{\ab^{3/2}}{\dst^2\priv_1^2}}$ samples from $\p$ and $\bigO{\frac{\ab^{3/2}}{\dst^2\priv_2^2}}$ samples from $\q$, as this protocol is $(\priv_1,0)-$LDP for samples in $\p$ and $(\priv_2,0)-$LDP private for samples in $\q$ respectively.
\end{theorem}

\begin{proof}
We claim~\cref{alg:local} below satisfies our demands. Its correctness directly follows from~\cref{theo:ldb_l1_to_l2_mean,eq:local:hadmard_pq_not_eq,fact:flip_1_bit,theo:l2_mean_flip}. 
\end{proof}
\begin{algorithm}[htp!]
\caption{Closeness testing under heterogeneous local differential privacy constraints}\label{alg:local}
\begin{algorithmic}[1]
\Require Privacy parameters $\priv_1, \priv_2$, a distance parameter $\dst$, $\ns_1$ users from unknown distribution $\p$ and $\ns_2$ users from unknown distribution $\q$.
\State Define $C_j = \{ i \in [K] : H_{ij}^{(K)} = 1 \}$, $j \in [K]$.
\State $\ns_1$ users from $\p$ and $\ns_2$ users from $\q$ are divided into $K$ disjoint subgroups of equal size separately. Users in the $j$th group generate a bit of message $1$ or $0$ depending on whether their data is in the set $C_J$.
\State Users from $\p$ flip their one bit of message with probability $\frac{1}{e^{\priv_1} + 1}$, and users from $\q$ flip their one bit of message with probability $\frac{1}{e^{\priv_2} + 1}$. Then users send their message to the analyser.
\State For users from $\p$, by taking one user from each block and viewing the resulting collection of messages as a length-$K$
binary vector, the analyser gets $\ns_1/K$ independent samples of a product-Bernoulli distribution $P \in \{ 0,1 \}^K$. 
\State The analyser does the same thing for users from $\q$ and gets $\ns_2/K$ independent samples of a product-Bernoulli distribution $Q \in \{ 0,1 \}^K$.
\State The analyser calculates $Z_2$ as defined in~\cref{eq:z2}. 
\If{$Z_2 \leq \frac{1}{2}\dst^2$}
    \State return accept.
\Else 
    \State return reject.
\EndIf
\end{algorithmic}
\end{algorithm}

\begin{remark}
    The sample complexity of this algorithm matches the known lower bound of sample complexity of locally private identity testing, and thus this algorithm is sample-optimal.
\end{remark}

\subsubsection{Public-coin protocol}
By combining the domain compression technique stated in~\cref{theo:random:dct:hashing} with~\cref{theo:ldp_ns}, we are then able to obtain a sample-efficient public-coin algorithm:
\begin{theorem}
    \label{theo:ldp_public_ns}
    For every $\ab \ge 0$ and $\priv_1, \priv_2 \in (0,1]$, there exist a private-coin protocol for $(\ab,\dst)$-closeness testing between two unknown distributions $\p,\,\q$ using $\bigO{\frac{\ab}{\dst^2\priv_1^2}}$
    samples from $\p$ and $\bigO{\frac{\ab}{\dst^2\priv_2^2}}$ samples from $\q$, as this protocol is $(\priv_1,0)-$LDP for samples in $\p$ and $(\priv_2,0)-$LDP private for samples in $\q$ respectively.
\end{theorem}
\begin{proof}
    Recall that by using the domain compression, we are able to compress the size of domain to $L$ while the total variation distance between any two distributions are preserved with high probability. Specifically, if we set the size of the compressed domain to be a constant $L = c_1$, we have
    a $(c_1, \sqrt{\frac{c_1c_2}{\ab}}, \delta)$-domain compression such that for all $\p, \q$ with domain sizes $\ab$ and total variation distance $\totalvardist{\p}{\q} > \dst$, the mapping satisfies
    \[
        \Pr[\totalvardist{\p^{\psi_U}}{\q^{\psi_U}} \geq \sqrt{\frac{c_1c_2}{\ab}} \dst] >  1 - \beta. 
    \]
    Also, $\p^{\psi_U} = \q^{\psi_U}$ when $\p=\q$.  
    
    From~\cref{theo:ldp_ns} we know that there is an algorithm testing whether $\p^{\psi_U} = \q^{\psi_U}$ or $\Pr[\totalvardist{\p^{\psi_U}}{\q^{\psi_U}} > \sqrt{\frac{c_1c_2}{\ab}} \dst]$ using $\bigO{\frac{c_1^{3/2}}{\sqrt{\frac{c_1c_2}{\ab}}^2\dst^2\priv_1^2}}$ samples from $\p$ and $\bigO{\frac{c_1^{3/2}}{\sqrt{\frac{c_1c_2}{\ab}}^2}}$ from $\q$ as this algorithm is $(\priv_1,0)-$LDP for samples in $\p$ and $(\priv_2,0)-$LDP private for samples in $\q$. 
    By standard probability amplification techniques, we can decrease the error probability by a constant factor by increasing the number of samples by a constant factor, to achieve any desired constant error probability $\beta_0$. That implies the probability of this algorithm outputting the correct answer is at least $ (1-c_3)(1-\beta_0)$. By choosing the right constants, the probability can be made larger than $\frac{2}{3}$.
\end{proof}

\subsection{Under Shuffle Privacy}
We now turn to our algorithms under the less stringent shuffle privacy model.
\subsection{Private-coin protocol}
For closeness testing under heterogeneous shuffle differential privacy constraints, we propose an algorithm based on Poisson mechanism, whose guarantees are stated in~\cref{lemma:poisson_mech}. (I.e., each user will add Poisson noise to their data to preserve their privacy in the shuffled model.) Since $1-e^{-x} \geq x/2$ for any $x \in [0,1]$, we only need to set $\mu = \bigO{\frac{\log (1/\delta)}{\priv^2}}$ when $\priv \in [0,1] ,\triangle \geq 1$, in this lemma.

We now claim their is a sample-efficient algorithm for closeness testing under heterogeneous shuffle differential privacy constraints, as stated in~\cref{theo:sdp_ns}.

\begin{theorem}
    \label{theo:sdp_ns}
    For every $\ab \ge 0$, $\priv_1, \priv_2 \in (0,1]$ (w.l.o.g. $\priv_1 \geq \priv_2$), and $\delta\in(0,1]$, there exist a private-coin protocol for $(\ab,\dst)$ closeness testing between two unknown distributions $\p,\,\q$ using 
    \[
    \ns_1 = \bigO{ \frac{\sqrt{k}}{\dst^2} + \frac{\ab^{3/4}\sqrt{\log(1/\delta)}}{\dst\priv_1}  + \min\Paren{\frac{\priv_1^2 \priv_2^2}{\dst^4 \log^2(1/\delta)}, \frac{\ab^{2/3}}{\dst^{4/3}}\Paren{\frac{\priv_2}{\priv_1}}^{2/3}}}
    \]
    samples from $\p$ and $\ns_2 = (\frac{\priv_1}{\priv_2})^2\ns_1$ samples from $\q$; and this protocol is $(\priv_1,\delta)-$shuffle differentially private for samples in $\p$ and $(\priv_2,\delta)-$shuffle differentially private for samples in $\q$ respectively.
\end{theorem}
\begin{remark}
    When $\priv_1=\priv_2$ and ignore $\delta$, our algorithm retrieves the best known upper bound even for identity testing in the shuffle differential privacy model.
\end{remark}

\begin{proof}[Proof of~\cref{theo:sdp_ns}]
We will use the so-called ``Poissonisation trick'' here: that is, instead of exactly $\ns_1$ and $\ns_2$, we will assume there are $\poisson{\ns_1}$ users getting samples from $\p$, and $\poisson{\ns_2}$ from $\q$.

Let each user use the Poisson mechanism to preserve their privacy. We end up with a histogram $X_1,X_2,...,X_{\ab}$ with Poisson noise from $\p$ and a histogram $Y_1,Y_2,...,Y_{\ab}$ with Poisson noise from $\q$. Specifically, we have $X_j \sim \poisson{\ns_1 \p_i + \mu_1}$ and $Y_j \sim \poisson{\ns_2 \q_i + \mu_2}$ for $j \in \{ 1,2,...,\ab \}$ where 
\begin{equation}
    \label{eq:poi_mu}
    \mu_1 = \bigO{\frac{\log (1/\privdelta)}{\priv_1^2}}, \mu_2 = \bigO{\frac{\log (1/\privdelta)}{\priv_2^2}}.
\end{equation}
Equivalently, we can rewrite this as
\[
        X_j \sim \poisson{ \nsa ((1-\gammaa)\p_j + \gammaa \uniform_j) }, \qquad Y_j \sim \poisson{ \nsb ((1-\gammab)\q_j + \gammab \uniform_j) }
\]
for $j\in[\ab]$, 
where $\uniform$ is the uniform distribution on $[\ab]$, and \begin{align}
    \label{eq:ns_gamma}
    \nsa &\eqdef \ns_1+\ab\mu_1\,, \quad \gammaa \eqdef \frac{\ab\mu_1}{\nsa}, \notag \\
    \nsb &\eqdef \ns_2+\ab\mu_2\,, \quad \gammab \eqdef \frac{\ab\mu_2}{\nsb}.
\end{align}

Let $\p' \eqdef (1-\gammaa)\p + \gammaa \uniform$ and $\q' \eqdef (1-\gammab)\q + \gammab \uniform$. From the above equations, we can see that our histograms are equivalent to taking $\nsa$ samples from $\p'$ and $\nsb$ samples from $\q'$ without any noise. (Note that $\nsa,\nsb$ are integers when $\mu_1,\mu_2$ are integers.) Moreover, setting $\gammaa = \gammab = \gamma$, we get $\totalvardist{\p'}{\q'} = (1-\gamma) \totalvardist{\p}{\q}$. Thus, we can test whether $\p = \q$ or $\totalvardist{\p}{\q} > \dst$ by testing whether $\p' = \q'$ or $\totalvardist{\p'}{\q'} > (1-\gamma)\dst$ given $N_1$ samples from $\p'$ and $N_2$ from $\q'$.

There exist several known sample-optimal algorithms for closeness testing from an unveven number of samples~\cite{ChanDVV14,BV:15,DK:16,DiakonikolasGKPP21}, showing one can test closeness of two distributions $\p,\,\q$ with distance parameter $\dst$ using $\ns_1 = \bigO{\frac{\sqrt{k}}{\dst^2} + \frac{k}{\dst^2 \sqrt{\ns_2}}}$ samples from $\p$ and $\ns_2$ samples from $\q$. By substituting $\ns_1,\ns_2$ with $\nsa,\nsb$ and $\dst$ with $(1-\gamma) \dst$, we have
\begin{equation}
 \nsa \gtrsim \frac{\sqrt{k}}{(1-\gamma)^2\dst^2} + \frac{k}{(1-\gamma)^2\dst^2 \sqrt{\nsb}}
 \label{eq:shuffle_eq_constr}
\end{equation}

Combined with~\cref{eq:poi_mu,eq:ns_gamma} and the setting $\gammaa = \gammab = \gamma$, we now show that the complexity bound will be satisfied as soon as
\begin{align*}
    \ns_1 &\gtrsim \frac{\sqrt{k}}{\dst^2} + \frac{\ab^{3/4}\sqrt{\mu_1}}{\dst}  + \min\Paren{\frac{1}{\dst^4\mu_1\mu_2}, \frac{\ab^{2/3}}{\dst^{4/3}}\Paren{\frac{\mu_1}{\mu_2}}^{1/3}},  \\
    \ns_2 &= \mleft(\frac{\priv_1}{\priv_2}\mright)^2\ns_1,
\end{align*}
which will conclude the proof. Details follows.
\begin{itemize}
    \item The first term of~\cref{eq:shuffle_eq_constr} imposes
$
 \nsa \gtrsim \frac{\sqrt{k}\nsa^2}{\ns_1^2\dst^2},
$
by using~\cref{eq:ns_gamma}; 
i.e.,
$
 \frac{\sqrt{k}}{\ns_1\dst^2} + \frac{\ab^{3/2}\mu_1}{\ns_1^2\dst^2} \lesssim 1 ,
$
which is equivalent to
\begin{equation}
    \label{eq:n1}
 \ns_1 \gtrsim \frac{\sqrt{k}}{\dst^2} \,\,\,\, \text{and} \, \,\,\, \ns_1 \gtrsim \frac{\ab^{3/4}\sqrt{\mu_1}}{\dst} \asymp \frac{\ab^{3/4}\sqrt{\log(1/\privdelta)}}{\dst\priv_1}.
\end{equation}
    \item Combined with~\cref{eq:ns_gamma}, the second term of~\cref{eq:shuffle_eq_constr} imposes
\begin{equation}
 \nsa \gtrsim \frac{\ab \nsa^2}{\ns_1^2\dst^2\sqrt{\nsb}}
 = \frac{\ab \nsb^2}{\ns_2^2\dst^2\sqrt{\nsb}}.
 \tag{By using~\cref{eq:ns_gamma}}
\end{equation}
$\nsa \gtrsim \frac{\ab \nsa^2}{\ns_1^2\dst^2\sqrt{\nsb}}$ implies
$
 \sqrt{\nsb} \gtrsim \frac{\ab^{3/2}\mu_1}{\ns_1^2\dst^2}+\frac{\ab}{\ns_1\dst^2},
$
i.e.,
\[
 \ns_2+\ab\mu_2 \gtrsim \frac{\ab^{3}\mu_1^2}{\ns_1^4\dst^4}+\frac{\ab^2}{\ns_1^2\dst^4}.
\]
And $\nsa \gtrsim \frac{\ab \nsb^2}{\ns_2^2\dst^2\sqrt{\nsb}}$ implies
\[
\ns_1+\ab\mu_1 \gtrsim \frac{\ab}{\dst^2\sqrt{\ns_2}} + \frac{\ab^{5/2}\mu_2^{3/2}}{\dst^2\ns_2^2} ,
\]
i.e.,
$
\frac{\ab}{\dst^2\sqrt{\ns_2}(n_1+k\mu_1)} + \frac{\ab^{5/2}\mu_2^{3/2}}{\dst^2\ns_2^2(n_1+k\mu_1)} \lesssim 1 ,
$
which leads to the sufficient condition
\[
\ns_2 \gtrsim \frac{\ab^2}{\dst^4(n_1+k\mu_1)^2} \,\,\,\, \text{and} \,\,\,\, \ns_2 \gtrsim \frac{\ab^{5/4}\mu_2^{3/4}}{\dst\sqrt{n_1+\ab\mu_1}},
\]
i.e.,
\begin{equation}
    \label{eq:n2:intermediate}
\ns_2 \gtrsim \min\Paren{ \frac{\ab^2}{\dst^4\ns_1^2}, \frac{\ab^2}{\dst^4\ab^2\mu_1^2} } \,\,\,\, \text{and} \,\,\,\, \ns_2 \gtrsim \min\Paren{ \frac{\ab^{5/4}\mu_2^{3/4}}{\dst\sqrt{n_1}},\frac{\ab^{5/4}\mu_2^{3/4}}{\dst\sqrt{\ab\mu_1}} },
\end{equation}
i.e.,
\begin{align*}
&\ns_2 \gtrsim \min\Paren{ \frac{\ab^2}{\dst^4}\frac{\dst^2{\priv_1}^2}{\ab^{3/2}\log(1/\privdelta)}, \frac{\ab^2}{\dst^4}\frac{\dst^4}{\ab}, \frac{{\priv_1}^4}{\dst^4{\log^2(1/\privdelta)}} } \\
\text{and} \,\,\,\ & \ns_2 \gtrsim \min\Paren{\frac{\ab^{5/4}}{\dst\sqrt{n_1}}\frac{\log^{3/4}(1/\privdelta)}{{\priv_2}^{3/2}},\frac{\ab^{3/4}\priv_1}{\dst\sqrt{\log(1/\privdelta)}}\frac{\log^{3/4}(1/\privdelta)}{{\priv_2}^{3/2}}},
\end{align*}
i.e.,
\begin{align*}
&\ns_2 \gtrsim \min(\frac{\ab^{1/2}{\priv_1}^2}{\dst^2\log(1/\privdelta)}, \ab, \frac{{\priv_1}^4}{\dst^4{\log^2(1/\privdelta)}}) \\ \text{and} \,\,\,\, &\ns_2 \gtrsim \min(\frac{\ab^{5/4}\log^{3/4}(1/\privdelta)}{\dst\sqrt{n_1}{\priv_2}^{3/2}},\frac{\ab^{3/4}\priv_1\log^{1/4}(1/\privdelta)}{\dst{\priv_2}^{3/2}}).
\end{align*}
\end{itemize}
Now, notice that our condition that $\gammaa=\gammab$ can equivalently be restated as 
\begin{equation}
    \label{eq:n1:n2}
        \ns_1 = \frac{\mu_1}{\mu_2}\ns_2 \asymp \Paren{\frac{\priv_2}{\priv_1}}^2 \ns_2
\end{equation}
(also, recall that we assumed wlog that $\priv_1 \geq \priv_2$, which implies $\mu_1 \leq \mu_2$). Then,~\cref{eq:n2:intermediate} becomes
\[
\ns_2 \gtrsim 
\min\Paren{ \frac{\ab^2\mu_2^2}{\dst^4\mu_1^2\ns_2^2}, \frac{1}{\dst^4\mu_1^2} } 
\,\,\,\, \text{and} \,\,\,\, \ns_2 \gtrsim \min\Paren{ \frac{\ab^{5/4}\mu_2^{5/4}}{\dst\sqrt{\mu_1}\sqrt{n_2}},\frac{\ab^{3/4}\mu_2^{3/4}}{\dst\sqrt{\mu_1}} }.
\]
This is equivalent to
\begin{equation}
    \label{eq:n2}
\ns_2 \gtrsim \min\Paren{ \frac{\ab^{2/3}}{\dst^{4/3}} \Paren{\frac{\mu_2}{\mu_1}}^{2/3}, \frac{1}{\dst^4\mu_1^2} } 
+ \min\Paren{ \frac{\ab^{5/6}\mu_2^{1/2}}{\dst^{2/3}}\Paren{\frac{\mu_2}{\mu_1}}^{1/3},\frac{\ab^{3/4}\mu_2^{1/4}}{\dst}\Paren{\frac{\mu_2}{\mu_1}}^{1/2} }\,.
\end{equation}
Replacing $\ns_2$ by its value $\frac{\mu_2}{\mu_1} \ns_1$ in ~\cref{eq:n2}, we then get
\[
\frac{\mu_2}{\mu_1} \ns_1 \gtrsim \min\Paren{ \frac{\ab^{2/3}}{\dst^{4/3}} \Paren{\frac{\mu_2}{\mu_1}}^{2/3}, \frac{1}{\dst^4\mu_1^2} } 
+ \min\Paren{ \frac{\ab^{5/6}\mu_2^{1/2}}{\dst^{2/3}}\Paren{\frac{\mu_2}{\mu_1}}^{1/3},\frac{\ab^{3/4}\mu_2^{1/4}}{\dst}\Paren{\frac{\mu_2}{\mu_1}}^{1/2} },
\]
i.e.
\begin{equation}
\label{eq:n1_2}
\ns_1 \gtrsim \min\Paren{ \frac{\ab^{2/3}}{\dst^{4/3}} \Paren{\frac{\mu_1}{\mu_2}}^{1/3}, \frac{1}{\dst^4\mu_1 \mu_2} }
+ \min\Paren{ \frac{\ab^{5/6}\mu_2^{1/2}}{\dst^{2/3}}\Paren{\frac{\mu_1}{\mu_2}}^{2/3},\frac{\ab^{3/4}\mu_2^{1/4}}{\dst}\Paren{\frac{\mu_1}{\mu_2}}^{1/2} }.
\end{equation}
By combining \cref{eq:n1,eq:n1_2}, we have
\begin{align*}
    \ns_1 &\gtrsim \frac{\sqrt{k}}{\dst^2} + \frac{\ab^{3/4}\sqrt{\mu_1}}{\dst}  + \min\Paren{\frac{1}{\dst^4\mu_1\mu_2}, \frac{\ab^{2/3}}{\dst^{4/3}}\Paren{\frac{\mu_1}{\mu_2}}^{1/3}}, \\ 
     \ns_2 &= \frac{\mu_2}{\mu_1}\ns_1.
\end{align*}
Combining this with~\cref{eq:poi_mu} yields the claimed bounds. This concludes the proof.
\end{proof}

\subsection{Public-coin protocol}
By leveraging again the domain compression technique of~\cref{theo:random:dct:hashing} with~\cref{theo:sdp_ns}, we get the following public-coin sample complexity:
\begin{theorem}
    \label{theo:sdp_public_ns}
    For every $\ab \ge 0$ and $\priv_1, \priv_2 \in (0,1]$, and $\delta\in(0,1]$, there exist a private coin protocol for $(\ab,\dst)$ closeness testing between two unknown distributions $\p,\,\q$ using 
    \[
    \ns_1 = \bigO{\frac{\sqrt{\ab}}{\dst^2} + \frac{\ab^{2/3}}{\dst^{4/3}\priv_1^{2/3}}\log^{1/3}\frac{1}{\delta} + \frac{\sqrt{\ab}}{\dst \priv_1} \sqrt{\log \frac{1}{\delta}}}
    \]
    samples from $\p$ and $\ns_2 = \bigO{\Paren{\frac{\priv_1}{\priv_2}}^2 \ns_1}$ samples from $\q$, as this protocol is $(\priv_1,\delta)-$shuffle differentially private for samples in $\p$ and $(\priv_2,\delta)-$shuffle differentially private for samples in $\q$ respectively.
\end{theorem}
\begin{proof}
    We also use domain compression technique, and the procedure is the same as that in the proof of~\cref{theo:ldp_public_ns}. The only thing we need to do is to choose an appropriate size of the compressed domain. I.e., we need to choose $2\leq L\leq \ab$ such that the following is minimised:
\[
    \ns_1 \gtrsim \frac{\sqrt{L}}{\Paren{\frac{\sqrt{L}\dst}{\sqrt{\ab}}}^2} + \frac{L^{3/4}\sqrt{\mu_1}}{\Paren{\frac{\sqrt{L}\dst}{\sqrt{\ab}}}}  + \min\Paren{\frac{1}{{\Paren{\frac{\sqrt{L}\dst}{\sqrt{\ab}}}}^4\mu_1\mu_2}, \frac{L^{2/3}}{{\Paren{\frac{\sqrt{L}\dst}{\sqrt{\ab}}}}^{4/3}}\Paren{\frac{\mu_1}{\mu_2}}^{1/3}},
\]
i.e.
\[
    \ns_1 \gtrsim \frac{\ab}{\sqrt{L}\dst^2} + \frac{L^{1/4}\sqrt{\ab}\sqrt{\mu_1}}{\dst}  + \min\Paren{\frac{\ab^2}{L^2{\dst}^4\mu_1\mu_2}, \frac{\ab^{2/3}}{\dst^{4/3}}\Paren{\frac{\mu_1}{\mu_2}}^{1/3}}.
\]

If we set $L$ to minimise the sum of the two terms, that is, 
\[
    L \eqdef \min\Paren{\ab, \max\Paren{2,\frac{\ab^{2/3}}{\dst^{4/3}\mu_1^{2/3}}}},
\]
we get
\[
    \ns_1 \gtrsim \frac{\sqrt{\ab}}{\dst^2} + \frac{\ab^{2/3}}{\dst^{4/3}}\mu_1^{1/3} + \frac{\sqrt{\ab}}{\dst}\sqrt{\mu_1}  + \min\Paren{\max\Paren{\frac{\ab^2}{\dst^4\mu_1\mu_2}, \frac{1}{{\dst}^4\mu_1\mu_2} , \frac{\ab^{2/3}}{\dst^{4/3}\mu_1^{2/3}\mu_2}}, \frac{\ab^{2/3}}{\dst^{4/3}}\Paren{\frac{\mu_1}{\mu_2}}^{1/3}}
\]
i.e.
\begin{equation}
    \ns_1 \gtrsim \frac{\sqrt{\ab}}{\dst^2} + \frac{\ab^{2/3}}{\dst^{4/3}}\mu_1^{1/3} + \frac{\sqrt{\ab}}{\dst}\sqrt{\mu_1}  + \min\Paren{\frac{\ab^2}{\dst^4\mu_1\mu_2}, \frac{\ab^{2/3}}{\dst^{4/3}}\Paren{\frac{\mu_1}{\mu_2}}^{1/3}}
    \label{eq:shuffle_before_rm_min}
\end{equation}

Since we have
\[
\min\Paren{\frac{\ab^2}{\dst^4\mu_1\mu_2}, \frac{\ab^{2/3}}{\dst^{4/3}}\Paren{\frac{\mu_1}{\mu_2}}^{1/3}} 
\leq  \frac{\ab^{2/3}}{\dst^{4/3}}\Paren{\frac{\mu_1}{\mu_2}}^{1/3}
\leq \frac{\ab^{2/3}}{\dst^{4/3}}\mu_1^{1/3},
\]

we can remove the last term in~\cref{eq:shuffle_before_rm_min}. Thus, the sample complexity is 
\begin{align*}
    \ns_1 &\gtrsim \frac{\sqrt{\ab}}{\dst^2} + \frac{\ab^{2/3}}{\dst^{4/3}\priv_1^{2/3}}\log^{1/3}\frac{1}{\delta} + \frac{\sqrt{\ab}}{\dst \priv_1} \sqrt{\log \frac{1}{\delta}}, \\ 
     \ns_2 &= (\frac{\priv_1}{\priv_2})^2\ns_1.
\end{align*}
This concludes the proof.
\end{proof}

\subsection{What About Central Privacy?}
    \label{ssec:central:dp}
We now recall our result for closeness testing in the heterogeneous central privacy model,~\cref{theo:central:intro}.
\begin{theorem}[Central Privacy, restated]
    \label{theo:central}
There exists an algorithm for closeness testing which guarantees $\priv_1$-differential privacy to the $\ns_1$ users of the first group, and $\priv_2$-differential privacy to the $\ns_2$ users of the second, with 
\[
    \ns_1 = \bigO{\max\Paren{\frac{\ab^{1/2}}{\dst^2}, \frac{\ab^{1/2}}{\priv_1^{1/2}\dst}, \frac{\ab^{2/3}}{\dst^{4/3}}, \frac{\ab^{1/3}}{\priv_1^{2/3}\dst^{4/3}}, \frac{1}{\priv_1\dst}}}
\]
and $\ns_2 = \bigO{\frac{\priv_1}{\priv_2}\ns_1}$ (assuming without loss of generality that $\priv_2 \leq \priv_1)$. 
\end{theorem}
As outlined in~\cref{sec:technique}, this follows straightforwardly by combining the result from~\cite{Zhang21} and privacy amplification by subsampling~\cref{theo:subsampling}.

This begs the question of whether this simple approach can be improved upon. We discuss below some other natural approaches, and why they failed or did not pan out.

    \subsubsection{Second idea: find a test statistics with heterogeneous sensitivity}

Designing a sample-efficient closeness testing algorithm under heterogeneous central differential privacy constraints is much more challenging than that under local and shuffle constraints. Our first attempt was to generalize previous work. As mentioned in the background section, a sample-optimal closeness testing algorithm under central differential privacy constraints was proposed in~\cite{Zhang21} by using the test statistic in~\cite{DiakonikolasGKPP21}. To be more specific, that test statistic is as follows. Suppose we take two sets of $\ns$ i.i.d.\ samples from both $\p$ and $\q$, and let $X_i, X'_i, Y_i, Y'_i$ be the number of occurrences of $i$ in those four sets respectively for $i \in [\ab]$ where $\ab$ is the size of the domain of $\p,\,\q$. Then the test statistic $Z$ is defined as
\begin{equation}
    \label{eq:statistic:z}
    Z \eqdef  |X_i-Y_i| + |X'_i-Y'_i| - |X_i-X'_i| - |Y_i-Y'_i| 
\end{equation}

In~\cite{Zhang21}, the algorithm shifts $Z$ to get a new test statistic 
\begin{equation}
    Z' = (Z - C_1\sqrt{\ns}- \frac{C_2}{\priv})/2
\end{equation}

Then, it uses a sigmoid function to map $\priv Z'$ to $(0,1)$ and then draws a Bernoulli random variable using this value as a parameter. Finally, the algorithm outputs "accept" or "reject" depending on whether the value of this random variable is $1$ or $0$. Roughly speaking, the idea is that Z will either be close to $0$ or greater than $\sqrt{n}$, depending on whether $\p=\q$
or $\totalvardist{\p}{\q}> \dst$; and therefore after shifting, $Z'$ will be either $< -C_2/\priv$ or $> C_2/\priv$ (with high probability). Using the sigmoid function on $\priv Z'$ maps this to a Bernoulli with bias either $1/2-\bigOmega{1}$ or $1/2+\bigOmega{1}$, which allows to distinguish the two cases while satisfying $\priv$-DP (since changing one sample will change the value of $Z$ by at most 2, as we will see below, and thus of $\priv Z'$ by at most $\priv$.

If we view $Z$ as a function of $\mathbf{X,X',Y,Y'}$ respectively, the $\ell_1$-sensitivity of $Z$ must be smaller than or equal to $2$. To see why, let us w.l.o.g. consider one sample in the first set taken from $\p$. The change of that sample will only add $1$ to $X_i$ and decrease $1$ from $X_j$ for some $i,j \in [\ab]$ where $i \neq j$. Thus, $Z$ will only by increased by $1$ or decreased by $1$ when that sample changes. By using a similar statement, one can show that the sensitivity of $\priv Z'$ is not larger than $\priv$. Combined with the fact $e^{- |\gamma|} \leq \sigma(x+\gamma) \sigma(x) \leq e^{|\gamma|}$, the author of~\cite{Zhang21} showed the algorithm is $(\priv,0)$-differential private in the central model.

Now, we consider how to generalize their algorithm. One natural idea is to adapt the test statistic to differentiate its sensitivity for samples from $\p$ and samples from $\q$. An example would be the following statistic:
\begin{equation} \label{eq:sta_example}
     Z = |\frac{X_i}{n}-\frac{Y_i}{m}|+|\frac{X'_i}{n}-\frac{Y'_i}{m}| -  |\frac{X_i}{n}-\frac{X'_i}{n}|-  |\frac{Y_i}{m}-\frac{Y'_i}{m}|.
\end{equation}

However, that test static is not easy to analyze as previous approaches do not naturally extend to this. Specifically, in~\cite{DiakonikolasGKPP21}, the analysis of the test static defined in~\cref{eq:statistic:z} relied on two things: (1)~the samples are drawn using the "Poissonisation trick" ($X_i, X'_i,Y_i, Y'_i$ are Poisson random variables), and (2) clever use of the additive property of Poisson distributions. Thus, we cannot use their proof technique for our proposed test statistic as a scaled Poisson (such as $X_i/\ns$) is no longer Poisson, and the different scalings prevent their approach from going through. To overcome that difficulty, we found a simpler and more general method to analyze the static defined in~\cref{eq:statistic:z}. It relies on an identity (Zolotarev identity) relating the expectation of the absolute value of any random variable to the integral of its characteristic function:
    \[
        \mathbb{E}{|X|} = \frac{2}{\pi}\int_0^\infty \frac{1-\Re(\mathbb{E}{e^{i tX}})}{t^2}\, dt
    \]

We have written the complete analysis in~\cite{CanonneS22}. However, we were still not able to find an appropriate tester even after finding the new analysis method. However, even this new analysis did not allow us to establish the desired properties of the testers we considered (such as the one in~\cref{eq:sta_example}), and it is
unclear whether they would actually work. We conjecture so, and proving this would be an interesting (and non-trivial) future direction.

\subsubsection{Third idea: using a different privatizing method}
In~\cite{Zhang21}, the closeness testing algorithm firstly calculates the test statistic using samples drawn from $\p$ and $\q$, then privatizes the test static. Instead, we want to privatize the histograms drawn from $\p$ and $\q$ respectively, and then use the privatized histogram to calculate a test statistic. 

However, we do not want to use a continuous noise such as the Laplace noise. If we use one of those continuous distributions for the noise, then analysing the statistic will become very hard. (E.g. the distribution of sums of Poisson/multinomial and Laplace can be very complicated.) For example, we can consider using the Skellam mechanism. (A Skellam distribution is the same as the distribution of the difference between two Poisson random variables.) The correctness of the Skellam mechanism directly follows from the correctness of the Poisson mechanism.

\begin{lemma}[Skellam noise]
    Let $f \colon X^n \rightarrow \mathcal{Z}$ be a $\triangle$-sensitive function. For any $\priv > 0, \delta \in (0,1)$ and $\lambda \geq \frac{16\log (10/\delta)}{(1-e^{-\priv/\triangle})^2} + \frac{2\triangle}{1-e^{-\priv/\triangle}}$, the randomized function $\mathcal{A}(X^n) = f(x) + Y_1 - Y_2$ where $Y_1,Y_2 \sim \poisson{\lambda}$ is $(\priv,\delta)-$differentially private in the shuffle model.
\end{lemma}
\begin{proof}
    Since the Poisson mechanism is correct, the randomized function $f$ is already $(\priv,\delta)-$differentially private after adding $Y_1$. Then it should still be $(\priv,\delta)-$differentially private after adding $Y_2$ due to the immunity of post-processing of differential privacy.
\end{proof}

Assuming we are dealing with homogeneous privacy constraints and want our output to be $(\priv,\delta)-$differentially private, we would set the parameter $\mu$ of the Skellam mechanism to be $\Theta(\log(1/\delta)/\priv^2)$. With some transformation, our test statistic could be written in this form:
 \[
        \tilde{Z} = \sum_{i=1}^{\ab} (|X_i + a_i - Y_i - c_i| + |X'_i + b_i - Y'_i - d_i| - |X_i + a_i' - X'_i - b_i'| - |Y_i + c_i' - Y'_i - d_i'|)
\]

where $a_1,b_i,c_i,d_i,a_1',b_i',c_i',d_i' \sim \poisson{2\mu}$ due to the additivity property of the Poisson distribution. By using the additivity property of the Poisson distribution again, it could be rewritten as
 \[
        \tilde{Z} = \sum_{i=1}^{\ab} (|A_i - B_i | + |A'_i - B'_i | - |A_i  - A'_i | - |B_i - B'_i |)
\]
where $A_i, A_i' \sim \poisson{\ns \p_i + 2\mu}$, $B_i, B_i' \sim \poisson{\ns \q_i + 2\mu}$. Since $\ns \p_i + 2\mu = (\ns + 2\ab \mu)\frac{\ns \p_i + 2\mu}{\ns + 2\ab \mu}$ and $\sum_{i=1}^{\ab}\frac{\ns \p_i + 2\mu}{\ns + 2\ab \mu} = 1$, using this test statistic is then equivalent to using the test statistic in~\cite{DiakonikolasGKPP21} by taking $\ns + 2\ab \mu$ samples from another distribution $\p'$ and $\ns + 2\ab \mu$ samples from another distribution $\q'$, where $\p'_i = \frac{\ns \p_i + 2\mu}{\ns + 2\ab \mu},\q'_i = \frac{\ns \q_i + 2\mu}{\ns + 2\ab \mu}$. Then, directly from the analysis in $\frac{\ns p_i + 2\mu}{\ns + 2\ab \mu}$ we have 
\[
        \mathbb{E}[\tilde{Z}]^2 \gtrsim  \frac{n^{3}\dst^4}{\ab(1+\ab\mu/\ns)}
\]

To bound the variance of the test statistic, we will use the Efron--Stein inequality. Our test statistic could be seen as the sum of $12\ab$ independent Poisson random variables. For each $i \in [\ab]$, there are 2 $\poisson{\ns \p_i}$ random variables, 2 $\poisson{\ns \q_i}$ random variables, 8 $\poisson{2\mu}$ random variables, and they are independent of each other. Then by using the Efron--Stein inequality, we get
     \[
        \var(\tilde{Z}) \le \frac{1}{2} \sum_{i=1}^{12\ab}\bE{}{(\tilde{Z} - \tilde{Z'_i})^2}
     \]

where the $i$th independent random variable in $\tilde{Z}$ is different from that in $\tilde{Z'_i}$, and other pairs of random variables are the same, respectively. This allows us to bound the variance as

\begin{align*} 
        \var(\tilde{Z}) \le \frac{1}{2}\sum_{i=1}^{\ab} 2 \cdot 2\ns p_i + 2 \cdot 2\ns q_i + 8 \cdot 2\cdot 2\mu = \bigO{\ns + \ab \mu}.
\end{align*} 

Then, for our tester to work, we need the "signal" to be larger than the "noise", i.e., we want $\mathbb{E}[\tilde{Z}]^2 \gg \var(\tilde{Z})$. When $\ns \ge \ab \mu$, this implies $\frac{n^{3}\dst^4}{\ab} \gtrsim \ns$, and thus $\ns = \bigOmega{\frac{k^{1/2}}{\dst^{2}}}$. Otherwise, we have $\frac{\ns^{3}\dst^4}{\ab(1+\ab\mu/\ns)} \gtrsim \ab \mu$, and thus $ \ns^3 \gtrsim  \frac{\ab^3 \mu}{n \dst^4}        \gtrsim \frac{\ab^3}{n \priv^3 \dst^4} $, which implies $\ns = \bigOmega{\frac{\ab^{3/2}}{\priv^{3/2} \dst^2} }$. Combining the results, we get the overall requirement that, for this tester to work, $\ns$ must be $\bigOmega{\frac{k^{1/2}}{\dst^{2}} + \frac{k^{3/4}}{\priv \dst} \sqrt{\log(1/\delta)}}$.

This bound is not optimal, but it is not too bad. However, it cannot be generalized to the heterogeneous setting because the privatizing procedure is, in fact, not separated (because it can be viewed as adding noise to the difference of histograms taken from two distributions). We also considered other methods of privatizing but were not able to develop an algorithm with good sample complexity due to the time frame of this thesis.

\subsubsection{Fourth idea: use "flattening samples" and "testing samples"}

In~\cite{DK:16}, the authors proposed a new approach for sample-optimally closeness testing. They firstly use some ``flattening samples'' from one distribution to ``flatten'' the domain in order to decrease the expected value of $\ell_2-$norm. Then, they use a standard $\ell_2$ tester to achieve $\ell_1$ closeness testing. It is natural to develop a close testing algorithm using different numbers of samples from two distributions, which is proposed in that paper. However, it is hard to make that algorithm differential private because the sensitivity of the ``flattening samples'' can be prohibitively large (as large as the number of samples). 

It is worth noting that the authors of~\cite{ADKR:19} developed a method for making the algorithm in~\cite{DK:16} deferentially private. They considered all permutations of ``flattening samples'' and ``testing samples'' in the calculation of the test statistic to decrease the sensitivity of ``flattening samples.'' However, we were not able to generalize their proof to our needs, and it is unclear whether this technique
could yield a sample-efficient tester for our
problem.
\section{Future Work}
    \label{sec:future:work}

Our results raise several future directions. First, it is not clear how to achieve pure privacy for the shuffle model, which is interesting since pure privacy is a stricter privacy guarantee. Second, while we are able to obtain sample-optimal algorithms for the local model tight bounds for the central and shuffle model remain unknown. Third, since our algorithms only work for the high privacy regime where $\priv \in (0,1]$, it will be interesting to determine algorithms for the low privacy regime. Note that a low-privacy algorithm for the local model should directly lead to a high-privacy algorithm for the shuffle model, \emph{via} amplification by shuffling. Last, it would be interesting to consider a mixed privacy guarantee, \eg where one group of users works a local model, while the other group relies on the shuffle model.
\printbibliography

\appendix
\section{Proof of Lemma ~\ref{theo:product_closeness_no_privacy}}
\label{app:proof:product_closeness_no_privacy}
\begin{proof}[Proof of~\cref{theo:product_closeness_no_privacy}]
We now prove the correctness of~\cref{alg:product_closeness_no_privacy} by using Chebyshev's inequality. Firstly, it is not hard to show the expectation of $Z_1$ is 
\begin{equation*}
\mathbb{E}(Z_1) = \begin{cases}
0 &\text{if $P = Q$,}\\
\norm{\mu(P) - \mu(Q)}^2_2 &\text{if $P \neq Q$.}
\end{cases}
\end{equation*}

To apply Chebyshev's inequality, we also need to calculate the variance of $Z_1$. Let $\triangle = \hat{X} - \hat{Y}, \triangle' = \hat{X'} - \hat{Y'},$ We have 
\begin{align*} 
    \mathbb{E}(Z_1^2) &= \mathbb{E}(\langle\, \triangle , \triangle'  \rangle^2) \\
    &= \mathbb{E}((\sum_{j=1}^{d}\triangle_j \triangle'_j)^2) \\
    &= \mathbb{E}(\sum_{j=1}^{d}((\triangle_j \triangle'_j)^2)) + \mathbb{E} (\sum_{1 \leq i,j \leq d, i \neq j} (\triangle_j \triangle'_j \triangle_i \triangle'_i)) \\
    &= \sum_{j=1}^{d}(\mathbb{E}^2(\triangle_j^2)) + \sum_{1 \leq i,j \leq d, i \neq j} \mathbb{E}^2(\triangle_j) \mathbb{E}^2(\triangle_i) \\
    &= \sum_{j=1}^{d}(\mathbb{E}^2(\triangle_j^2)) + \sum_{j=1}^{d} \mathbb{E}^2(\triangle_j) \sum_{i=1}^{d}\mathbb{E}^2(\triangle_i) - \sum_{j=1}^{d}(\mathbb{E}^4(\triangle_j))\\
    &= \sum_{j=1}^{d}(\mathbb{E}^2(\triangle_j^2)) + \norm{\mu(P)-\mu(Q)}_2^4 - \norm{\mu(P)-\mu(Q)}_4^4
\end{align*}

For any $1 \leq j \leq d$, we have
\begin{align*}
    \mathbb{E}(\triangle_j^2) &= \mathbb{E}^2(\triangle_j) + \var(\triangle_j) \\
    &= \mathbb{E}^2(\triangle_j) + \var(\hat{X}_j) + \var(\hat{Y}_j) \\
    &= \mathbb{E}^2(\triangle_j) + \frac{1}{\ns}\mu(P)_i(1-\mu(P)_i) + \frac{1}{\ns}\mu(Q)_i(1-\mu(Q)_i) \\
    &\leq (\mu(P)_i - \mu(Q)_i)^2 + \frac{1}{\ns}.
\end{align*}

It follows that
\[
    \sum_{j=1}^{d}(\mathbb{E}^2(\triangle_j^2)) \leq \norm{\mu(P) - \mu(Q)}_4^4 + \frac{2\norm{\mu(P) - \mu(Q)}^2_2}{\ns} + \frac{d}{\ns^2}
\]
Thus,
\begin{align*}
    \var(Z_1) &= E(Z_1^2) - E^2(Z_1) \\
    &= \sum_{j=1}^{d}(\mathbb{E}^2(\triangle_j^2)) + \norm{\mu(P)-\mu(Q)}_2^4 - \norm{\mu(P)-\mu(Q)}_4^4 - \norm{\mu(P)-\mu(Q)}^4_2 \\
    &\leq \norm{\mu(P) - \mu(Q)}_4^4  + \frac{d}{\ns^2} + \frac{2\norm{\mu(P) - \mu(Q)}^2_2}{\ns}  - \norm{\mu(P)-\mu(Q)}_4^4  \\
    &= \frac{d}{\ns^2}  + \frac{2\norm{\mu(P) - \mu(Q)}^2_2}{\ns} \\
    &= \frac{\dst^2}{100^2}  + \frac{\dst^2 \norm{\mu(P) - \mu(Q)}^2_2 }{50 \sqrt{d}}
\end{align*}

When $P = Q$, $\sqrt{\var(Z_1)} = \frac{\dst}{100}$. By Chebyshev's inequality, we have 
\begin{align}
    \bPr{Z_1 > \frac{\dst^2}{2}} < \bPr{Z_1 \geq \sqrt{3}\sqrt{\var(Z_1)}}   < \frac{1}{3}.
    \label{eq:closeness_product_no_noise_1}
\end{align}

When $\norm{\mu(P) - \mu(Q)}_2 > \dst$, we have $\sqrt{\var(Z_1)} < \frac{\dst}{100}  + \frac{\dst \norm{\mu(P) - \mu(Q)}_2 }{\sqrt{50}} < \frac{11 \norm{\mu(P) - \mu(Q)}^2_2 }{100}$. By Chebyshev's inequality, we have 
\begin{align}
    \bPr{Z_1 < \frac{\dst^2}{2}} < \bPr{|Z_1 - \norm{\mu(P) - \mu(Q)}^2_2| \geq \sqrt{3}\sqrt{\var(Z_1)}} < \frac{1}{3}.
    \label{eq:closeness_product_no_noise_2}
\end{align}
 
By combining~\cref{eq:closeness_product_no_noise_1,eq:closeness_product_no_noise_2}, we conclude that this algorithm always return the correct answer with probability at least $\frac{2}{3}$.
\end{proof}

\section{Proof of Theorem~\ref{theo:l2_mean_flip}}
\label{app:proof:l2_mean_flip}
\begin{proof}[Proof of~\cref{theo:l2_mean_flip}]
    We draw two groups of samples $X^{(1)},...,X^{(\ns)},X'^{(1)},...,X'^{(\ns)}$ from the $d-$dimensional product distribution $P$ with noise and two groups of samples $Y^{(1)},...,Y^{(\ns)},Y'^{(1)},...,Y'^{(\ns)}$ from the $d-$dimensional product distribution $Q$ with noise, where $\q,\,\q \in \{ -1,\,1 \}^d$, $\ns_1 = \frac{800\sqrt{d}}{\priv_1^2 \dst^2} $ and $\ns_2 = \frac{800\sqrt{d}}{\priv_2^2 \dst^2} $. Then, we calculate the following test statistic
\begin{equation}
    Z_2 = \langle\, \frac{e^{\priv_1}+1}{e^{\priv_1}-1}(\hat{X} - \frac{1}{e^{\priv_1}+1}) - \frac{e^{\priv_2}+1}{e^{\priv_2}-1}(\hat{Y} - \frac{1}{e^{\priv_2}+1}) , \, \frac{e^{\priv_1}+1}{e^{\priv_1}-1}(\hat{X'} - \frac{1}{e^{\priv_1}+1}) - \frac{e^{\priv_2}+1}{e^{\priv_2}-1}(\hat{Y'} - \frac{1}{e^{\priv_2}+1}) \rangle
    \label{eq:z2}
\end{equation}
where $\hat{X}, \hat{X'}, \hat{Y}, \hat{Y'}$ are the mean vectors of $X,X',Y,Y'$ respectively. We accept $P = Q$ if $Z_2 \leq \frac{1}{2}\dst^2$, and reject otherwise.

Here, $\frac{e^{\priv_1}+1}{e^{\priv_1}-1}(\hat{X} - \frac{1}{e^{\priv_1}+1}), \frac{e^{\priv_1}+1}{e^{\priv_1}-1}(\hat{X'} - \frac{1}{e^{\priv_1}+1})$ are unbiased estimators for $\mu(P)$ while $\frac{e^{\priv_2}+1}{e^{\priv_2}-1}(\hat{Y} - \frac{1}{e^{\priv_2}+1})$ and $\frac{e^{\priv_2}+1}{e^{\priv_2}-1}(\hat{Y'} - \frac{1}{e^{\priv_2}+1})$ are unbiased estimators for $\mu(Q)$. Thus, the expectation of $Z_2$ is the same as $Z_1$. To apply Chebyshev's inequality for this test statistic, we again calculate the variance of $Z_2$.

Let $\gamma_1,\gamma_2$ to be $\frac{e^{\priv_1}+1}{e^{\priv_1}-1},\,\frac{e^{\priv_2}+1}{e^{\priv_2}-1}$ respectively, then our test statistic can be written as 
\[
    Z_2 = \langle\, \gamma_1(\hat{X} - \frac{1}{2}(1-\frac{1}{\gamma_1})) - \gamma_2(\hat{Y} - \frac{1}{2}(1-\frac{1}{\gamma_2})) , \, \gamma_1(\hat{X'} - \frac{1}{2}(1-\frac{1}{\gamma_1})) - \gamma_2(\hat{Y'} - \frac{1}{2}(1-\frac{1}{\gamma_2})) \rangle.
\]
I.e.,
\[
    Z_2 = \langle\, \gamma_1(\hat{X} - \frac{1}{2}) - \gamma_2(\hat{Y} - \frac{1}{2}) , \, \gamma_1(\hat{X'} - \frac{1}{2}) - \gamma_2(\hat{Y'} - \frac{1}{2}) \rangle.
\]
Letting $\triangle = \gamma_1(\hat{X} - \frac{1}{2}(1-\frac{1}{\gamma_1})) - \gamma_2(\hat{Y} - \frac{1}{2}(1-\frac{1}{\gamma_2}))$, we have 
\begin{align*} 
    \mathbb{E}(Z_2^2) = \sum_{j=1}^{d}(\mathbb{E}^2(\triangle_j^2)) + \norm{\mu(P)-\mu(Q)}_2^4 - \norm{\mu(P)-\mu(Q)}_4^4.
\end{align*}
And 
\begin{align*}
    \mathbb{E}(\triangle_j^2) &= \Paren{\mathbb{E}^2(\triangle_j) + \var(\triangle_j)}^2\\
    &= \Paren{(\mu(P)_j - \mu(Q)_j)^2 + \var(\hat{X}_j)) + \var(\hat{Y}_j)}^2 \\
    &= \Paren{(\mu(P)_j - \mu(Q)_j)^2 + \frac{\gamma_1^2}{\ns_1}\mu(P)_j(1-\mu(P)_j) + \frac{\gamma_2^2}{\ns_2}\mu(Q)_j(1-\mu(Q)_j)}^2 \\
    &\leq \Paren{(\mu(P)_j - \mu(Q)_j)^2 + \frac{\gamma_1^2}{\ns_1}\mu(P)_j + \frac{\gamma_2^2}{\ns_2}\mu(Q)_j}^2 \\
    &= (\mu(P)_j - \mu(Q)_j)^4 + 2(\frac{\gamma_1^2}{\ns_1}\mu(P)_j + \frac{\gamma_2^2}{\ns_2}\mu(Q)_j)(\mu(P)_j - \mu(Q)_j)^2 + (\frac{\gamma_1^2}{\ns_1}\mu(P)_j + \frac{\gamma_2^2}{\ns_2}\mu(Q)_j)^2 \\
    &\leq (\mu(P)_j - \mu(Q)_j)^4 + 2(\frac{\gamma_1^2}{\ns_1} + \frac{\gamma_2^2}{\ns_2})(\mu(P)_j - \mu(Q)_j)^2 + (\frac{\gamma_1^2}{\ns_1} + \frac{\gamma_2^2}{\ns_2})^2
\end{align*}

Thus, 
\begin{align*} 
    \mathbb{E}(Z_2^2) &= \sum_{j=1}^{d}(\mathbb{E}^2(\triangle_j^2)) + \norm{\mu(P)-\mu(Q)}_2^4 - \norm{\mu(P)-\mu(Q)}_4^4 \\
    &\leq \norm{\mu(P)-\mu(Q)}_2^4 + 2\mleft(\frac{\gamma_1^2}{\ns_1} + \frac{\gamma_2^2}{\ns_2}\mright)\norm{\mu(P)-\mu(Q)}_2^2 + d\mleft(\frac{\gamma_1^2}{\ns_1} + \frac{\gamma_2^2}{\ns_2}\mright)^2.
\end{align*}

Since $\frac{e^{x}+1}{e^{x}-1} < \frac{3}{x}$ for any $x \in (0, 1 ]$, we have 
\begin{align*}
    \var(Z_2) &= \mathbb{E}(Z_2^2) - \mathbb{E}^2(Z_2) \\
    &\leq 2\mleft(\frac{\gamma_1^2}{\ns_1} + \frac{\gamma_2^2}{\ns_2}\mright)\norm{\mu(P)-\mu(Q)}_2^2 + d\mleft(\frac{\gamma_1^2}{\ns_1} + \frac{\gamma_2^2}{\ns_2}\mright)^2 \\
    &< 6\mleft(\frac{1}{\priv_1^2 \ns_1} + \frac{1}{\priv_2^2 \ns_2}\mright)\norm{\mu(P)-\mu(Q)}_2^2 + 3d\mleft(\frac{1}{\priv_1^2 \ns_1} + \frac{1}{\priv_2^2 \ns_2}\mright)^2 \\
    &= \frac{3 \dst^2 \norm{\mu(P)-\mu(Q)}_2^2}{200 \sqrt{d}} + \frac{12 \dst^4}{800^2}.
\end{align*}

When $P = Q$, $\sqrt{\var(Z_1)} = \frac{\sqrt{3}}{400} \dst$. By Chebyshev's inequality, we have 
\begin{align}
    \bPr{Z_2 > \frac{\dst^2}{2}} < \bPr{Z_1 \geq \sqrt{3}\sqrt{\var(Z_1)}}  < \frac{1}{3}.
    \label{eq:closeness_product_with_noise_1}
\end{align}

When $\norm{\mu(P) - \mu(Q)}_2 > \dst$, we have $\sqrt{\var(Z_1)} <  \sqrt{\frac{ 3}{200} + \frac{12}{800^2}}\norm{\mu(P) - \mu(Q)}^2_2 $. By Chebyshev's inequality, we have 
\begin{align}
    \bPr{Z_1 < \frac{\dst^2}{2}} < \bPr{|Z_1 - \norm{\mu(P) - \mu(Q)}^2_2| \geq \sqrt{3}\sqrt{\var(Z_1)}} < \frac{1}{3}.    
    \label{eq:closeness_product_with_noise_2}
\end{align}

By combining~\cref{eq:closeness_product_with_noise_1,eq:closeness_product_with_noise_2}, we conclude that this algorithm always return the correct answer with probability at least $\frac{2}{3}$.
\end{proof}

\end{document}